\documentclass{article}
\usepackage[margin=1in]{geometry}
\usepackage[hidelinks]{hyperref}
\usepackage{amsmath}
\usepackage{amssymb}
\usepackage{parskip} 
\usepackage{natbib}
\usepackage{mathrsfs}
\usepackage{pdfpages}
\usepackage{bm}
\usepackage{booktabs}
\usepackage{graphics}
\usepackage{multicol}
\usepackage[bf]{caption}
\usepackage{multirow}
\usepackage{subcaption}
\usepackage{floatrow}
\usepackage{amsthm}

\newtheorem{theorem}{Theorem}[section]

\newtheorem{lemma}[theorem]{Lemma}
\usepackage{color}

\usepackage{setspace}
\usepackage{mwe}  
\usepackage{authblk}

\title{Multi-porous extension of anisotropic poroelasticity: linkage with micromechanics
}

\author[1, 2]{\normalsize Filip P. Adamus}
\author[1]{\normalsize David Healy}
\author[2]{\normalsize Philip G. Meredith}
\author[2]{\normalsize Thomas M. Mitchell}
\affil[1]{\footnotesize {School of Geosciences, University of Aberdeen, Aberdeen, UK}}
\affil[2]{\footnotesize {Department of Earth Sciences, University College London, London, UK \\ \linebreak
\[adamusfp@gmail.com \quad d.healy@abdn.ac.uk \quad p.meredith@ucl.ac.uk \quad tom.mitchell@ucl.ac.uk\]}}

\date{}

\begin{document}
\maketitle
\begin{abstract}
We attempt to formalise the relationship between the poroelasticity theory and the effective medium theory of micromechanics. 
The assumptions of these two approaches vary, but both can be linked by considering the undrained response of a material; and that is the main focus of the paper. 
To analyse the linkage between poroelasticity and micromechanics, we do not limit ourselves to the original theory of Biot.
Instead, we consider a multi-porous extension of anisotropic poroelasticity, where pore fluid pressure may vary within the bulk medium of interest.
As a consequence, any inhomogeneities in the material are not necessarily interconnected; instead, they may form isolated pore sets that are described by different poroelastic parameters and fluid pressures.
We attempt to incorporate the effective methods inside Biot-like theory and investigate the poroelastic response of various microstructures.
We show the cases where such implementation is valid and the others that appear to be questionable. 
During micromechanical analysis, we derive a particular case of cylindrical transverse isotropy---commonly assumed in conventional laboratory triaxial tests---where the symmetry is induced by sets of aligned cracks.
\end{abstract}


\section{Introduction}
The theory of poroelasticity describes the coupling between the deformation in a solid porous framework (or matrix) and the changes in fluid pressure or content residing in the pores or cracks. The fundamental equations were derived by Biot in a series of papers describing the consolidation of porous materials~\citep{Biot41,Biot56,Biot62}, although the name ``poroelasticity'' was first used by~\citet{Geertsma57}. The formal definition of poroelasticity relies on the assumption of statistically homogeneous continua (i.e., the ergodic hypothesis): in particular, a single connected solid phase comprising the matrix and a single connected pore space containing the fluid. 

Poroelasticity is important because most rocks, especially in the accessible upper crust of the Earth, spend most of their life cycle in the poroelastic regime: i.e., fluid-saturated and stressed, contracting and expanding in response to natural or man-made forces. As we enter the energy transition to deal with the climate emergency, it is imperative that we have a thorough understanding of how poroelasticity works, at scales varying from grains, pores and cracks to whole reservoirs and fault zones. This understanding will help us deliver the sequestration of carbon dioxide, extract geothermal energy and store hydrogen beneath the surface in a safe and cost-efficient way.

Isotropic poroelasticity describes the case where the single connected pore space has no preferred orientation. This has proven useful for describing the deformational response of porous granular rocks, such as sandstones, to changes in load or fluid pressure~\citep[e.g.,][]{Hart95}. However, the assumption of isotropy is inappropriate, when the pore space is made up of aligned cracks---a common feature in fractured rocks, especially around tectonic fault zones, and under conditions of differential rather than hydrostatic loading---then we need the equations describing the anisotropic behaviour~\citep{Cocco02, Lockner02}. Even though the theoretical basis for describing anisotropic poroelasticity is well established, detailed experimental verification at the laboratory scale remains relatively rare, and the isotropic assumption is often misused~\citep[e.g.,][]{Beeler00}.

Another issue hinges on the time and length scales of fluid movement in the pore space, and whether it is useful (or necessary) to consider the pore space as a single connected domain, and whether there is a constant pore fluid pressure throughout. Consider the example shown in Figure~\ref{fig:pic}, a granular cemented sandstone deformed to brittle failure in the laboratory. The pore space (shown in black) can be characterized as (at least) two distinct domains: narrow preferentially aligned cracks and more equant (although irregularly-shaped) pores. The critical questions are: how connected are these domains of void space (porosity clusters); and if they are connected, what is the time scale of fluid flow between them? One method is to consider two end-member possibilities: firstly, that pores and cracks are fully connected, and there is a single pore fluid pressure throughout~\citep{Biot41}; and secondly, that they are not connected at all and that pore fluid pressure varies at the scale of grains and pores and cracks~\citep{Kachanov80}. However, there are scenarios that lie between these end-member possibilities. For example, as depicted in Figure~\ref{fig:scheme}, a medium can consist of distinct porosity clusters (pore sets) that are isolated or weakly connected. In this case, the fluid pressure is constant within a pore set, but may vary between pore sets throughout the material. Such a view underlines the multi-porous (or multiple-porosity) generalisation of poroelasticity that was proposed by e.g.,~\citet{Berryman02},~\citet{Mehrabian14},~\citet{Mehrabian18} for isotropy and hydrostatic confining pressure. Since aligned cracks may play a significant role in fluid pressure variation, anisotropic expressions involving differential stress (instead of hydrostatic pressure) are necessary. In this paper, we consider fundamental expressions of multi-porous extension of anisotropic poroelasticity that describe the deformation of a medium containing different pore and crack sets, each with distinct fluid pressure. Such expressions are derived and discussed in more detail in our parallel paper dedicated solely to the ``extended poroelasticity''~\citep{AdamusEtAl23a}.
\begin{figure}[!htbp]
\centering
\begin{subfigure}{.43\textwidth}
\includegraphics[scale=0.44]{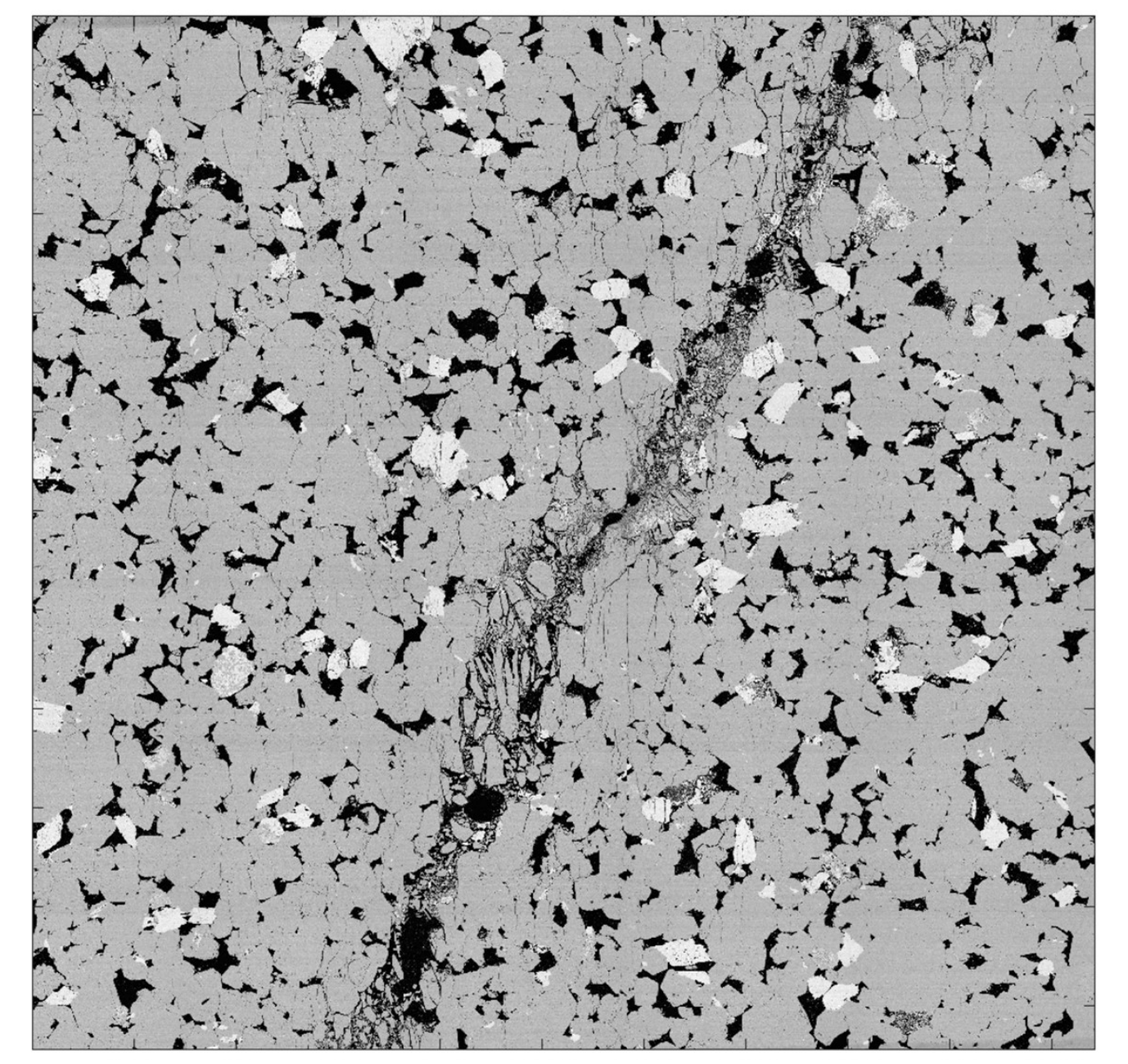}
\caption{}
\label{fig:pic}
\end{subfigure}
\qquad\qquad
\begin{subfigure}{.4\textwidth}
\includegraphics[scale=0.433]{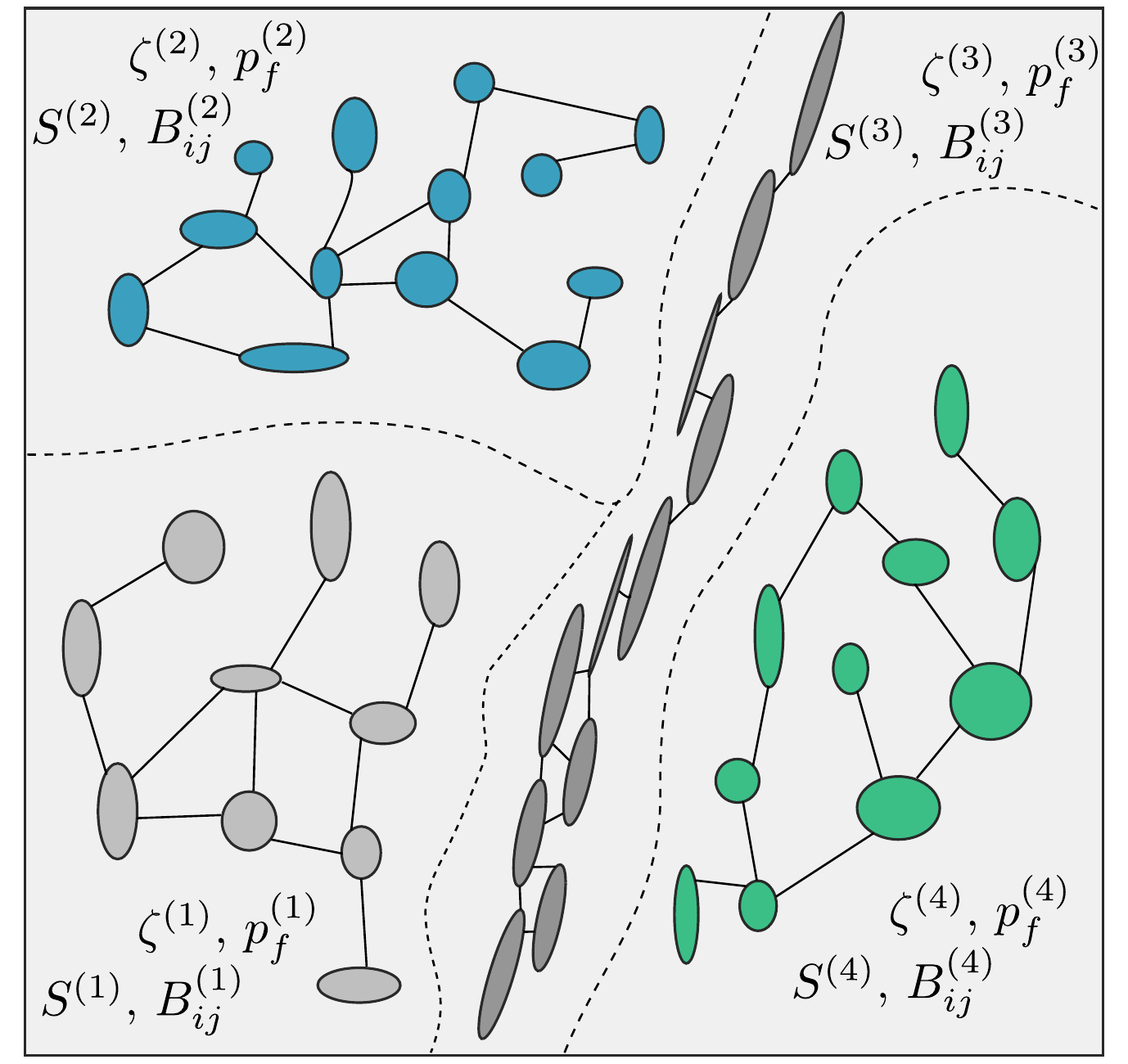}
\caption{}
\label{fig:scheme}
\end{subfigure}
\caption{\footnotesize{\textbf{(a)} SEM-BSE image of a faulted Hopeman sandstone specimen from a laboratory failure test~\citep{Rizzo18}. The field of view is approximately 20 mm across. Black areas are void space, light grey is feldspar, and mid-grey is quartz. The void space can be characterized as two types: thin cracks, mostly oriented parallel to the vertical axis in this image and concentrated around the through-going fault surface; and more or less equant pores distributed in the regions around the fault. \textbf{(b)} Schematic view of the extended poroelasticity. Each pore set is described by distinct fluid content change, pore pressure, storage, and Skempton-like coefficients. Distinct colours correspond to solid matrix and isolated sets.} }
\label{fig:picscheme}
\end{figure}

Further, a common approach in the past has been to rely on approximations from effective elasticity (effective medium theory, EMT) either as a direct modelling approach or in the interpretation of laboratory experimental results~\citep[e.g.,][]{Wong17}. EMT takes a microstructural approach to describe the pores and cracks and their impact on the bulk material properties. However, the formal relationship between poroelasticity (sensu stricto; Biot) and effective elasticity remains unclear. As noticed by~\citet{Gueguen09}, there is an essential difference between the holistic approach of original poroelasticity (connected porosity) and the individual approach of micromechanics (many isolated pores). This corresponds to set-impact and pore-impact descriptions, respectively, seen from the perspective of the extended Biot theory. In other words, a static bulk medium (macroscopic scale), where distinct sets (mesoscopic scale) contain individual pores (microscopic scale), can be viewed at either a set scale or pore scale. The question arises whether (and, if so; when) it is possible to combine the EMT with poroelastic expressions. In the case of many isolated pores (e.g., cracks),~\citet{Shafiro97} define the pore fluid pressure ``polarisation''; the phenomenon that appears in the effective medium approach and corresponds to different pressures in each pore: such a situation cannot formally be considered as poroelastic, according to Biot theory.  
Therefore, in this paper, we analyse the poroelastic extension in view of micromechanics, thereby bridging the gap between the original continuum-based analyses of~\citet{Biot41} and the microstructural, isolated pore and crack models of~\citet{Kach18}. In contrast to~\citet{DormieuxEtAl02}, we assume uniform stress boundary conditions and utilise the concept that excess compliance is the superposition of a dry pore and fluid phase impacts~\citep{Shafiro97}. We consider various microstructures for which the effective approach can be strictly or approximately valid within the poroelastic extension. 

Last but not least, within the micromechanical analysis of extended poroelasticity, we consider a particular case of effective transverse isotropy (TI).
Specifically, we focus on the symmetry that is induced by TI-oriented penny-shaped cracks embedded in the solid matrix~\citep{Sayers95}.
In the context of laboratory experiments, the aforementioned effective case is often referred to as cylindrical transverse isotropy (CTI). 
Naturally, the TI-oriented cracks can be perceived as vertical (axial) cracks that are random when viewed in the horizontal (radial) plane. 
However, we show that TI-oriented cracks can also be analogous to vertically (axially) aligned cracks forming sets (either connected or isolated) that are equally distributed around the vertical symmetry axis.
Such a novel representation allows us to describe each set geometry and permits the implementation of the extended poroelasticity (various pressures) within CTI. 

Let us introduce some notions that are used throughout the paper repeatedly. 
We define a pore set or a set as connected inhomogeneities that may allow fluid flow.  
Also, we define a pore subset as connected inhomogeneities that make part of a pore set or are equal to a set.
Further, we refer to a pore group as multiple inhomogeneities that are not necessarily connected.
Finally, the notion ``identical pores'' denotes inhomogeneities of the same geometry; thus, having identical shape, orientation, and size.
Due to the large number of symbols used in the paper, a full list can be found in Appendix~\ref{sec:list}.
\section{Extended poroelasticity in view of micromechanics}
Consider a medium containing connected pores, where fluid can flow. 
Assume that certain connected inhomogeneities may be isolated from the others forming distinct sets of pores $(p)$. 
To describe the deformation of such a poroelastic medium, we propose equations governing the strains of the entire porous material ($\varepsilon_{ij}$) and the change of fluid content in each set ($\zeta^{(p)}$), respectively. Taking into account the effect of $n$ different sets, we get
\begin{equation}\label{one}
\varepsilon_{ij}=\sum_{k=1}^3\sum_{\ell=1}^3S_{ijk\ell}\sigma_{k\ell}+\frac{1}{3}\sum_{p=1}^nS^{(p)}B^{(p)}_{ij}p_f^{(p)}\,,
\end{equation}
\begin{equation}\label{two}
\zeta^{(p)}=\frac{1}{3}S^{(p)}\sum_{k=1}^3\sum_{\ell=1}^3B^{(p)}_{k\ell}\sigma_{k\ell}+S^{(p)}p_f^{(p)}\,,
\end{equation}
where $S_{ijk\ell}$ denotes compliance of a porous skeleton and $\sigma_{k\ell}$ is the remote, uniform stress applied to the medium. A particular set of pores is described by $S^{(p)}$, $B^{(p)}_{ij}$, and $p_f^{(p)}$ that stand for a storage coefficient, Skempton-like second-rank tensor, and pore pressure, respectively. Throughout the paper, $i,\,j\in\{1,\,2,\,3\}$. Following~\citet{Biot41} convention, pressure has the opposite sign as compared to stress.

Let us discuss the above expressions. 
They are designed to account for various $n$ sets of pores having any microstructure (shape, orientation, and size). 
As mentioned earlier, pores within a particular set are connected to each other; however, different sets are treated as isolated.  
Therefore, fluid cannot flow between such defined sets. 
Each porosity must be considered individually since it produces a particular fluid content change $\zeta^{(p)}$; when summed, giving total fluid content change in the bulk volume, $\zeta_{tot}=\sum_{p=1}^n\zeta^{(p)}$.
As a consequence of isolated sets, fluid pressure is not necessarily constant either.
It may vary if the microstructure of each pore set differs, which is analogous to the pressure polarisation effect~\citep{Shafiro97}.
In turn, varying pressure affects storage and Skempton-like coefficients that need to be calculated for each set separately. 
This fact comes from the definition of the aforementioned parameters, where strict relation to fluid pressure is apparent~\citep{Cheng97}.
Similarly to the change of fluid content, the storage coefficients can also be summed to obtain the total storage of the bulk volume, $S_{tot}=\sum_{p=1}^nS^{(p)}$.
On the other hand, such a summation does not make sense in the context of set pressures or Skempton-like tensors.
One should treat them as poroelastic characteristics of each set, and nothing more.
Using the analogy of a stratified medium, it makes sense to sum the volume fractions of voids or thicknesses of constituents, but adding the elasticity tensors or densities of layers is rather pointless.
The aforementioned properties of pore pressure or storage and Skempton-like coefficients are explained further in Appendix~\ref{sec:ap}.
As expected, in the case of a single set of pores ($n=1$), expressions (\ref{one})--(\ref{two}) reduce to the original Biot theory designed for constant fluid pressure, single fluid content change, one storage coefficient, and one Skempton tensor~\citep{Biot41,Biot62,Cheng97}.

From the perspective of micromechanical linkage with (extended) poroelasticity, expression~(\ref{one}) is critical and needs to be analysed further. In micromechanics, fluid flow is not considered and the strain-stress relation, analogous to~(\ref{one}), is provided only. In other words, the inhomogeneity is treated either as dry or saturated. Therefore, in the context of the linkage between both theories, expression~(\ref{two}) may seem to be redundant.
However, as will become more clear shortly, a specific, undrained ($\zeta^{(p)}=0$) version of expression~(\ref{two}) is necessary for comparison of the theories.
Except for drained (dry) or undrained end-member cases of poroelasticity, other scenarios do not have the analogy to micromechanics.
Therefore, due to no extra value in view of the direct poroelasticity-micromechanics linkage, the analysis of intermediate states or time dependency is beyond the scope of this paper. Note that the time-dependent multi-porous extension of anisotropic poroelasticity is discussed in our parallel article~\citep{AdamusEtAl23a}.

Note that (\ref{one})--(\ref{two}), in contrast to the derivations of e.g.,~\citet{Mehrabian18}, allow material to be anisotropic and do not assume confining pressure. The generalisation to anisotropy becomes crucial when coping with nonrandom pores or aligned cracks---common geological scenarios---that, in turn, may lead to variable pore pressure.
Further, it is important that expression~(\ref{one}) is also relevant to likely scenarios, where pore-sets are not strictly isolated but possess certain weak connections among each other. Yet, these connections are considered to be weak enough so that pressure in each set can vary (e.g, dual porosity in gas reservoir). To account for the set connections and time factor leading to eventual pore pressure equilibration, additional coupling terms in expression~(\ref{two}) should appear~\citep{AdamusEtAl23a}. Although, as noticed by researchers working on an isotropic extension of poroelasticity~\citep{Berryman02,Mehrabian14,Mehrabian18}, these coupling terms are very small and can therefore be neglected. (Also, the low permeability of interconnections implies a very long time required for pressure equilibration).
The introduction of the coupling terms would lead to unwanted complications of the micromechanical analysis; therefore, they are not invoked herein. Nevertheless, our derivations can be treated as approximately valid for the above-mentioned weakly-connected pore sets.

Let us consider two limiting cases that simplify expressions~(\ref{one})--(\ref{two}) and allow us to grasp the physics contained in them.
In the ideal set-drained conditions, where for every set $p_f^{(p)}=0$, expression~(\ref{one}) reduces to
\begin{equation}\label{two:b}
\varepsilon_{ij}=\sum_{k=1}^3\sum_{\ell=1}^3S_{ijk\ell}\sigma_{k\ell}\,.
\end{equation}
Hence, the physical meaning of the above compliance tensor is the following.
It describes the elastic properties of an effective medium containing drained sets of pores. 
Besides, such a medium may also contain some closed spaces either dry or filled with fluid.
Therefore, $S_{ijk\ell}$ denotes compliances of a set-drained, but not necessarily dry, medium.
In other words, pore pressure $p_f^{(p)}$ from expressions~(\ref{one})--(\ref{two}) corresponds to the pores that are able to be drained only.
Pressure in closed pores is implicitly included in the stress tensor from expression~(\ref{two:b}).

In the case of undrained conditions, where for every set, $\zeta^{(p)}=0$, expression~(\ref{two}) reduces to
\begin{equation}\label{three}
p_f^{(p)}=-\frac{1}{3}\sum_{k=1}^3\sum_{\ell=1}^3B^{(p)}_{k\ell}\sigma_{k\ell}\,.
\end{equation}
Upon inserting it inside expression~(\ref{one}), we obtain compliances of the undrained effective medium,
\begin{equation}\label{four}
\varepsilon_{ij}=\sum_{k=1}^3\sum_{\ell=1}^3\left(S_{ijk\ell}-\frac{1}{9}\sum^n_{p=1}S^{(p)}B^{(p)}_{ij}B^{(p)}_{k\ell}\right)\sigma_{k\ell}=\sum_{k=1}^3\sum_{\ell=1}^3S^{u}_{ijk\ell}\sigma_{k\ell}\,.
\end{equation}
Thus, $S^{u}_{ijk\ell}$, stand for undrained compliances. 
The effect of fluids contained in pore sets corresponds to 
\begin{equation}\label{delta}
\Delta_{ijk\ell}:=\sum_{p=1}^n\Delta_{ijk\ell}^{(p)}=-\frac{1}{9}\sum_{p=1}^nS^{(p)}B^{(p)}_{ij}B_{k\ell}^{(p)}\,.
\end{equation}
In Figure~\ref{fig:one}, we depict the impact of tensors $\Delta^{(p)}$ on a set-drained porous medium. 
Therein, we exemplify possible microstructures to which our theoretical extension is pertinent. 
From now on, the notion of ``fluid effect'' refers to tensor $\Delta$ or $\Delta^{(p)}$, depending on the context.
In the next sections, the end-member strain-stress relations (3) and (5) are compared with their micromechanical strain-stress counterparts.
\begin{figure}[!htbp]
\centering
\begin{subfigure}{.4\textwidth}
  \centering
\includegraphics[scale=0.33]{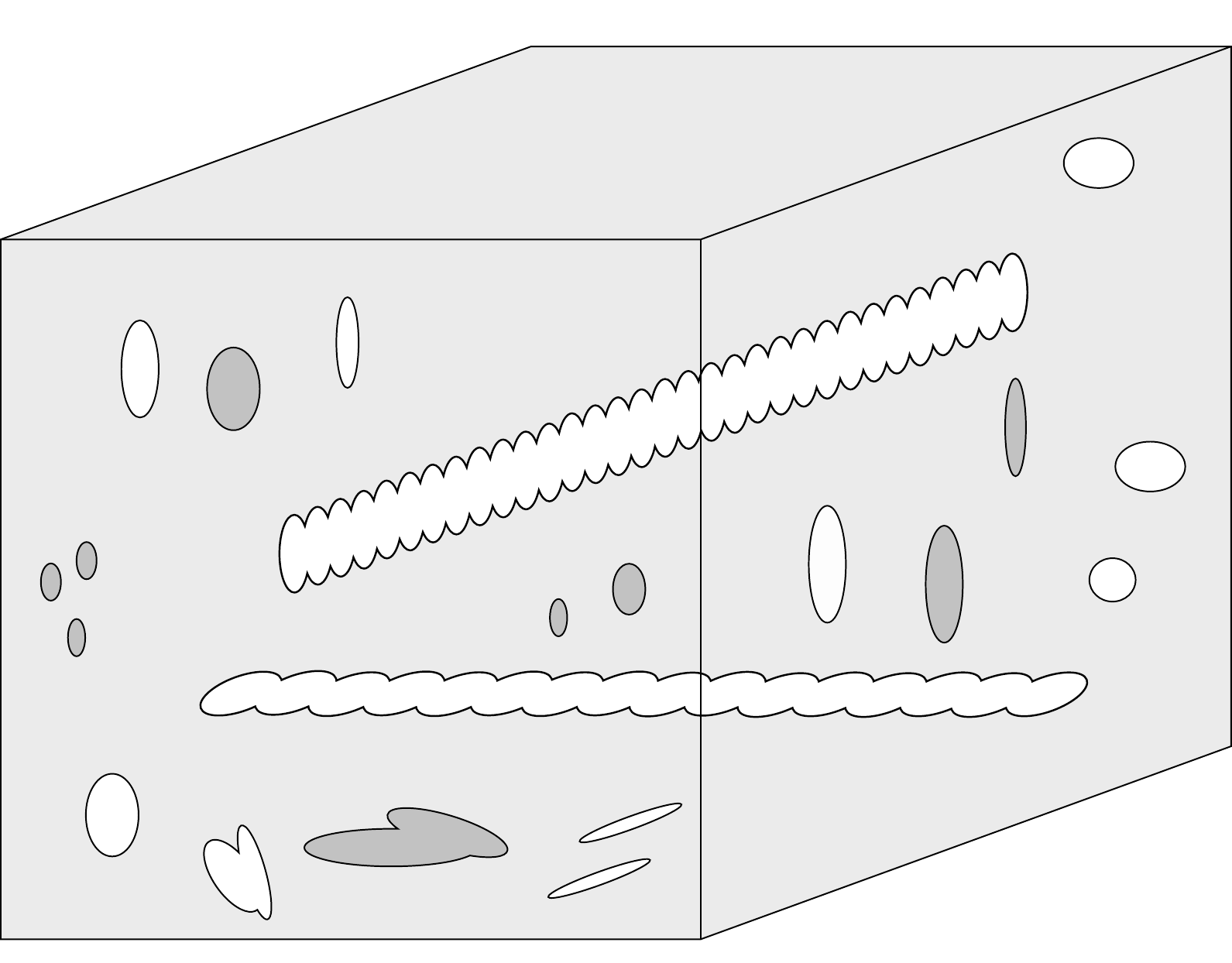}
\caption{\footnotesize{Set-drained medium ($p_f^{(1)}=p_f^{(2)}=0$)}}
\end{subfigure}
\qquad\qquad
\begin{subfigure}{.4\textwidth}
  \centering
\includegraphics[scale=0.33]{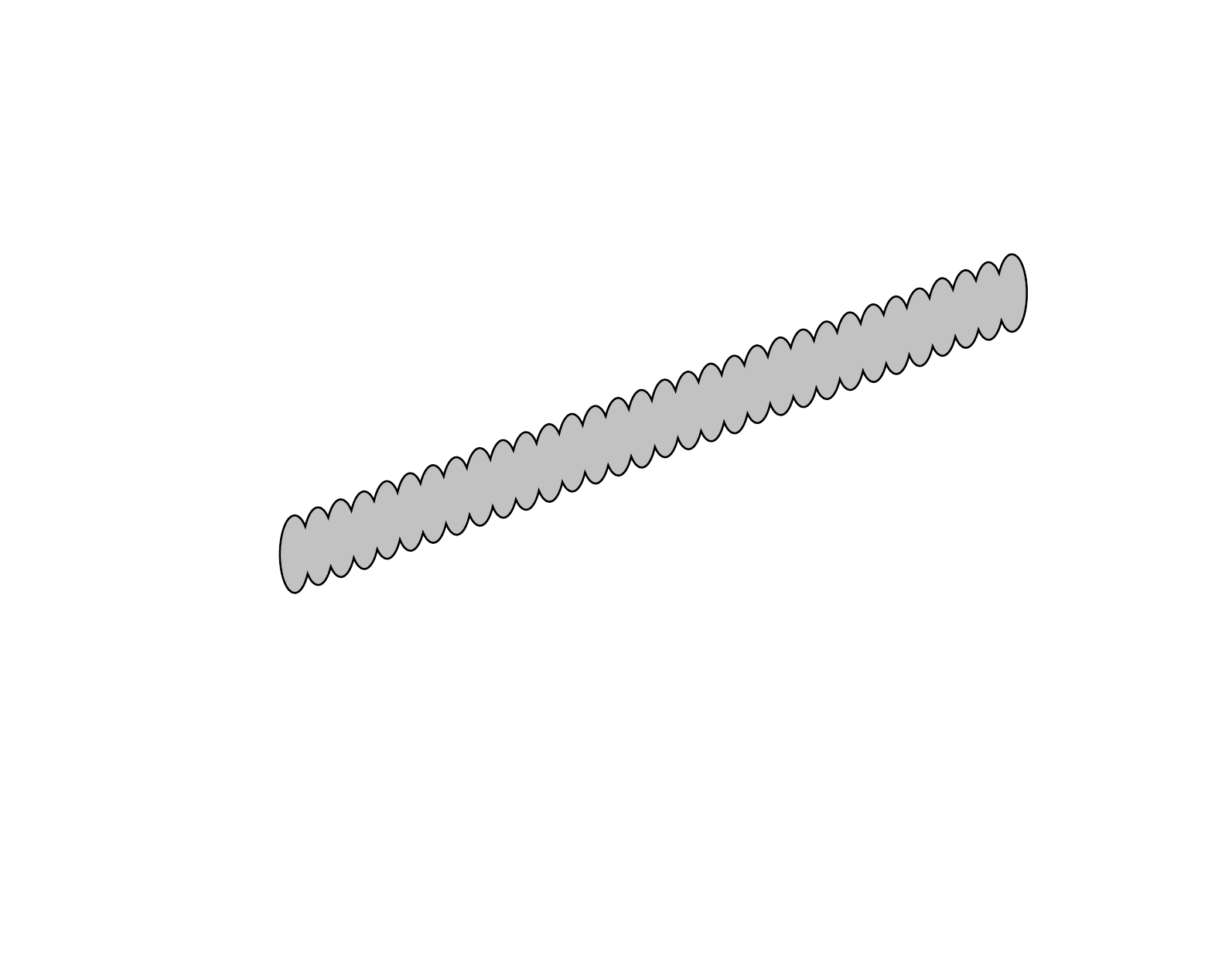}
\caption{\footnotesize{First set ($\Delta^{(1)}$, where $\zeta^{(1)}=0$)}}
\end{subfigure}
\qquad\qquad
\begin{subfigure}{.4\textwidth}
  \centering
\includegraphics[scale=0.33]{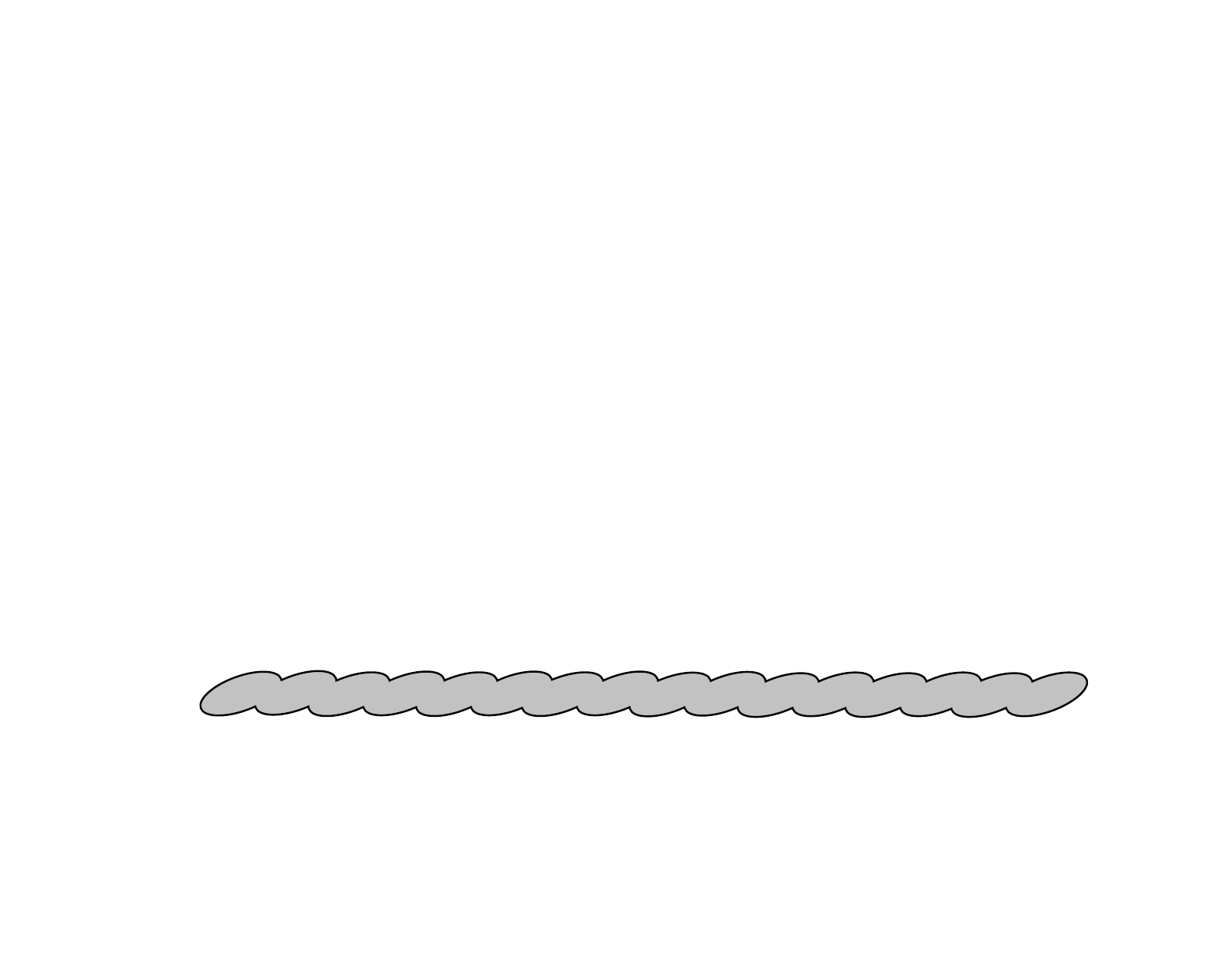}
\caption{\footnotesize{Second set ($\Delta^{(2)}$, where $\zeta^{(2)}=0$)}}
\end{subfigure}
\qquad\qquad
\begin{subfigure}{.4\textwidth}
  \centering
\includegraphics[scale=0.33]{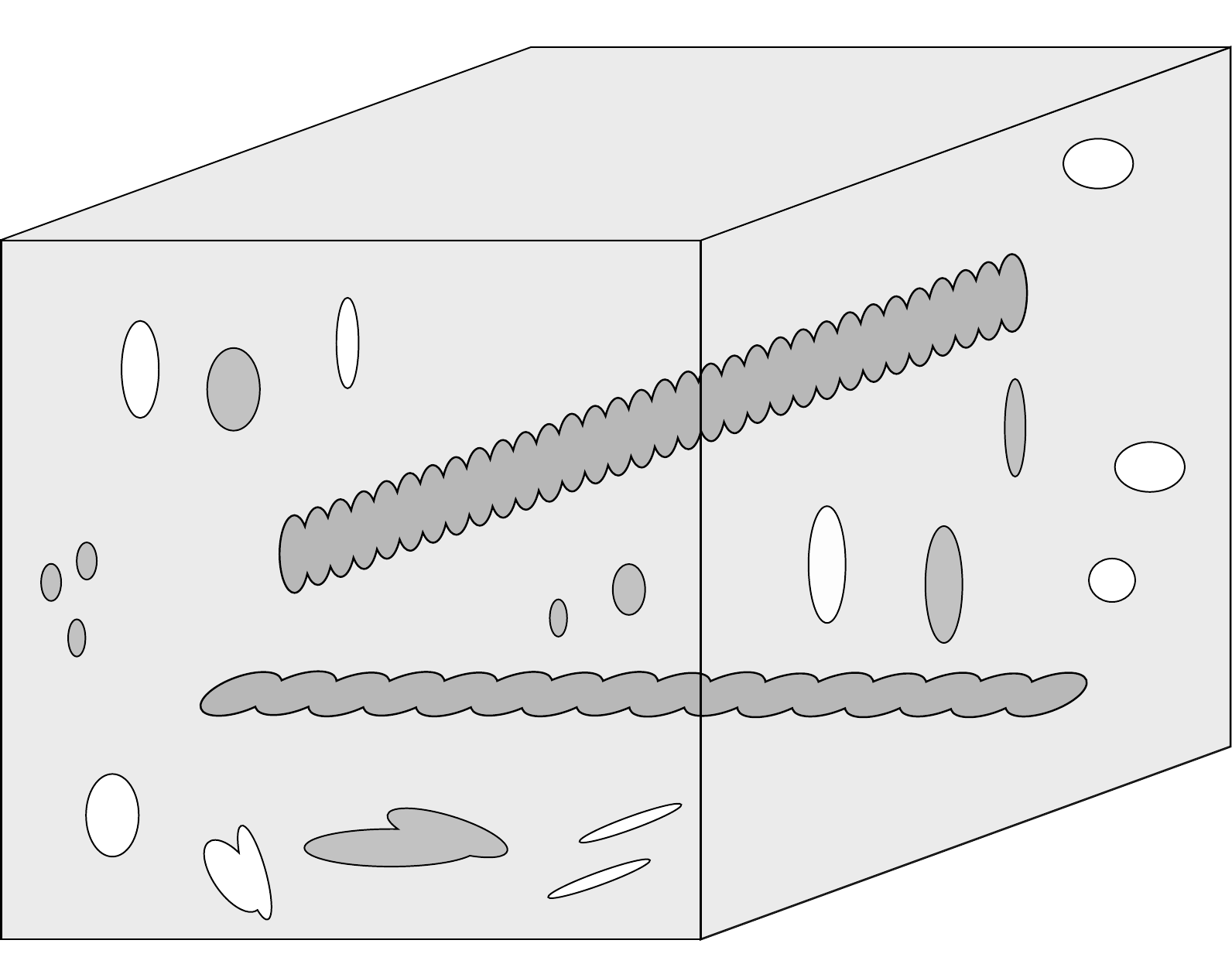}
\caption{\footnotesize{Undrained medium ($\zeta^{(1)}=\zeta^{(2)}=0$)}}
\end{subfigure}
\caption{\footnotesize{Illustration of the theoretical considerations. Light grey stands for a solid, dark grey for fluid, white colour denotes empty space. To consider the influence of fluid on compliances, we insert two sets of undrained pores \textbf{(b)} and \textbf{(c)} inside set-drained skeleton \textbf{(a)}. This way the compliances of a skeleton, $S_{ijk\ell}$, change due to the addition of tensors $\Delta^{(1)}$ and $\Delta^{(2)}$.
As a consequence, undrained compliances, $S_{ijk\ell}^u$, are obtained \textbf{(d)}.
In our schematic drawings, the skeleton contains dry and saturated closed spaces. Also, pores in each set have an identical geometry and overlap slightly. }}
\label{fig:one}
\end{figure}
%
\subsection{Micromechanical analysis: unspecified microstructure}\label{sec:21}
Let us perform a micromechanical analysis to get even more insight into expressions~(\ref{one})--(\ref{two}).
We want to translate the poroelastic (storage and Skempton-like) coefficients into compliances.
The micromechanical description of expressions~(\ref{one})--(\ref{two}) may be practical in the context of the reproducibility of laboratory measurements.
First, we consider a single saturated pore embedded in the solid matrix---viewed by EMT as a representative volume element (REV)---being much larger than the pore size. 
Second, we analyse sets of identical pores embedded in the same REV.
Finally, a more general case of pores of different shapes and orientations in sets is discussed.
In this section, we do not specify the microstructure but indicate only whether pores are identical.
\subsubsection{Single undrained pore}
Let us refer to the effective method proposed by \citet{Shafiro97} or \citet{Kach18} that was designed for undrained inhomogeneities.
For a single undrained pore, the aforementioned micromechanical researchers propose
\begin{equation}\label{e_ij_kach}
\varepsilon_{ij}=\sum_{k=1}^3\sum_{\ell=1}^3\left[S^{0}_{ijk\ell}+\phi\left(H_{ijk\ell}+\Delta H_{ijk\ell}\right)\right]\sigma_{k\ell}\,,
\end{equation}
\begin{equation}\label{polar}
p_f=\sum_{k=1}^3\sum_{\ell=1}^3Q_{k\ell}\sigma_{k\ell}\,,
\end{equation}
where $S^{0}_{ijk\ell}$ are the compliances of the solid phase, $H_{ijk\ell}$ are the excess compliances caused by a dry pore, $\Delta H_{ijk\ell}$ are the excess compliances caused by the fluid in the pore, $\phi$ is the volume fraction occupied by the pore, and $Q_{ij}$ denotes components of the fluid polarisation tensor. 
Naturally, in this case, $\phi$ is also equal to the total volume fraction occupied by all pores, $\phi_{\rm{tot}}$.
Components $\Delta H_{ijk\ell}$---to which we refer loosely as ``saturated compliances''---are expressed in terms of dry excess compliances and the fluid polarisation tensor, namely,
\begin{equation}\label{DeltaH}
\Delta H_{ijk\ell}=\sum^3_{m=1}H_{ijmm}Q_{k\ell}\,,
\end{equation}
where 
\begin{equation}\label{q}
Q_{ij}=-\frac{K_d}{1+\delta}\,\sum^3_{m=1}H_{ijmm}\,,
\end{equation}
and
\begin{equation}\label{Kc}
K_d=\left(\sum^3_{m=1}\sum^3_{n=1}H_{mmnn}\right)^{-1}\,
\end{equation}
denotes the bulk modulus of the dry pore.
$\delta$ is a factor introduced by \citet{OConnell74} and generalised by \citet{Shafiro97}, namely,
\begin{equation}\label{delta_small}
\delta=\frac{\frac{1}{K_f}-\frac{1}{K_0}}{\frac{1}{K_d}}\,
\end{equation}
where $K_f$ is the fluid bulk modulus and $K_0$ is the bulk modulus of the solid phase.
Upon inserting expressions~(\ref{q})--(\ref{delta_small}) into (\ref{DeltaH}), we get
\begin{equation}\label{DeltaH2}
\Delta H_{ijk\ell}=-\left(\frac{1}{K_d}+\frac{1}{K_f}-\frac{1}{K_0}\right)\,Q_{ij}Q_{k\ell}\,.
\end{equation}
For the case a dry pore, components $\Delta H_{ijk\ell}$ are equal to zero.
\subsubsection{Undrained sets with identical pores}
Let us now consider a group of $m_1$ undrained pores (either connected or isolated) embedded in a solid matrix. 
Assuming negligible interactions between the pores, we utilise expression~(\ref{e_ij_kach}) to get
\begin{equation}\label{multiple_identical}
\varepsilon_{ij}=\sum_{k=1}^3\sum_{\ell=1}^3\left[S^{0}_{ijk\ell}+\sum_{c=1}^{m_1}\phi_c \left(H_{ijk\ell_c}+\Delta H_{ijk\ell_c}\right)\right]\sigma_{k\ell}\,,
\end{equation}
where subscript $c$ is introduced to distinguish each pore in the group.
If the pores are identical, then the volume fraction and the excess compliances are the same for each pore. In such a case, subscript $c$ is no longer needed and~(\ref{multiple_identical}) reduces to
\begin{equation}\label{multiple_identical2}
\varepsilon_{ij}=\sum_{k=1}^3\sum_{\ell=1}^3\left[S^{0}_{ijk\ell}+m_1\phi \left(H_{ijk\ell}+\Delta H_{ijk\ell}\right)\right]\sigma_{k\ell}\,,
\end{equation}
where the total volume fraction of pores, $\phi_{\rm{tot}}=m_1\phi$.
Hence, it does not matter whether each identical pore is treated separately or the entire group is considered as a single inhomogeneity; both approaches are equivalent due to $\sum_{c=1}^{m_1}\phi_c=m_1\phi_c=m_1\phi$ and $H_{ijk\ell_c}=H_{ijk\ell}\implies\Delta H_{ijk\ell_c}=\Delta H_{ijk\ell}$.
We refer to this special case of equivalence in the below text repeatedly.
Note that~(\ref{multiple_identical2}) is essentially the same as~(\ref{e_ij_kach}), only the value of $\phi_{\rm{tot}}$ differs.

In the context of effective methods that use the non-interactive approximation discussed by~\citet{Kach18}, the mechanical response of a medium is not affected by possible connections between pores. 
Each inhomogeneity in a group is regarded individually and described in a manner that does not allow a distinction between connected or isolated pores.
This comes from the fact that~\citet{Kach18} do not consider fluid content changes; with inhomogeneities being treated as undrained. 
However, in the context of the poroelasticity theory, fluid content can change so that a clear distinction between pores that contain fluid that is stuck or allowed to flow is necessary. 

In regards to the paragraph above, we distinguish the compliances of isolated pores from the compliances that account for pores where fluid content varies.
Hence, we divide the inhomogeneities into a group with isolated pores ($m_0$) and a set with connected pores ($m=m_1-m_0$), to obtain
\begin{equation}\label{multiple_identical3}
\varepsilon_{ij}=\sum_{k=1}^3\sum_{\ell=1}^3\left(S_{ijk\ell}+m\phi\Delta H_{ijk\ell}\right)\sigma_{k\ell}\,,
\end{equation}
where 
\begin{equation}\label{Sijkl}
S_{ijk\ell}=S^{0}_{ijk\ell}+m_0\phi \left(H_{ijk\ell}+\Delta H_{ijk\ell}\right)+m\phi H_{ijk\ell}\,.
\end{equation}
If all pores are isolated, then $m=0$ and we obtain expression~(\ref{two:b}).
If all pores are connected, then $m=m_1$.
Note that $m\phi$ corresponds to the volume fraction occupied by a single interconnected pore set, $\phi^{(p)}$.
Similarly, $H_{ijk\ell}=H^{(p)}_{ijk\ell}$ that implies $\Delta H_{ijk\ell}=\Delta H^{(p)}_{ijk\ell}$, $Q_{ij}=Q_{ij}^{(p)}$, and $K_d=K_d^{(p)}$.
Inserting~(\ref{DeltaH2}) into~(\ref{multiple_identical3}), we obtain strains in a medium containing dry or saturated closed pores and a single interconnected set of identical pores, namely,
\begin{equation}\label{multi_identical}
\varepsilon_{ij}=\sum_{k=1}^3\sum_{\ell=1}^3\left[S_{ijk\ell}-\phi^{(p)}\left(\frac{1}{K^{(p)}_d}+\frac{1}{K_f}-\frac{1}{K_0}\right)\,Q^{(p)}_{ij}Q^{(p)}_{k\ell}\right]\sigma_{k\ell}\,.
\end{equation}
Term
\begin{equation}\label{def:storage}
\phi^{(p)}\left(\frac{1}{K^{(p)}_d}+\frac{1}{K_f}-\frac{1}{K_0}\right)=:S^{(p)}
\end{equation}
is a definition of the storage coefficient~\citep{Cheng97}.
Also, comparing expressions~(\ref{three}) and~(\ref{polar}) for $n=1$, we notice that the fluid polarisation tensor is related to the Skempton-like tensor,
\begin{equation}
Q^{(p)}_{ij}=-\frac{1}{3}B^{(p)}_{ij}\,,
\end{equation}
where
\begin{equation}\label{def:skempton}
B_{ij}^{(p)}:=\frac{3\phi^{(p)}}{S^{(p)}}\sum^3_{m=1}H^{(p)}_{ijmm}\,.
\end{equation}
Hence, if $n=1$, expression~(\ref{multi_identical}) is equivalent to (\ref{four}). 
In other words, the fluid effect caused by a set of identical pores,
\begin{equation}\label{equivalence}
\Delta^{(p)}_{ijk\ell}\equiv\phi^{(p)}\Delta H^{(p)}_{ijk\ell}=\sum_{c=1}^{m}\phi_c\Delta H_{ijk\ell_c}\,,
\end{equation}
can be expressed in terms of either poroelastic~(\ref{delta}) or elastic constants~(\ref{equivalence}).
The equivalence (\ref{equivalence}) and definitions~(\ref{def:storage}) and~(\ref{def:skempton}) are also valid for $n>1$ sets.
Considering dry and saturated excess compliances separately for each set, we get
\begin{equation}\label{ident}
\varepsilon_{ij}=\sum_{k=1}^3\sum_{\ell=1}^3\left(S_{ijk\ell}+\sum_{p=1}^n\phi^{(p)}\Delta H_{ijk\ell}^{(p)}\right)\sigma_{k\ell}=\sum_{k=1}^3\sum_{\ell=1}^3\left(S_{ijk\ell}+\sum_{p=1}^n\Delta_{ijk\ell}^{(p)}\right)\sigma_{k\ell}\,,
\end{equation}
where
\begin{equation}\label{ident2}
S_{ijk\ell}=S^{0}_{ijk\ell}+\sum_{c=1}^{n_0}\phi^{(c)} \left(H^{(c)}_{ijk\ell}+\Delta H^{(c)}_{ijk\ell}\right)+\sum_{p=1}^n\phi^{(p)}H^{(p)}_{ijk\ell}\,.
\end{equation}
The above expressions are the multiple-set generalisations of expressions~(\ref{multiple_identical3})--(\ref{Sijkl}), where isolated pores of the same microstructure are denoted by a superscript ($c$).
Herein, we allow $H_{ijk\ell}^{(c)}\neq H_{ijk\ell}^{(p)}$, $\Delta H_{ijk\ell}^{(c)}\neq \Delta H_{ijk\ell}^{(p)}$, and $n_0\neq n$. 
Note that expressions~(\ref{ident})--(\ref{ident2}) correspond to the scenario depicted in Figure~\ref{fig:one} (where $n_0\neq 0$ and $n=2$).
\subsubsection{Undrained sets with non-identical pores}
The micromechanical analysis of interconnected pores having different geometries is not straightforward; even if we again assume no interactions between inhomogeneities.
We propose two possible micromechanical descriptions, depicted in Figure~\ref{fig:three}.
One way is to consider each pore separately (at the microscopic scale) and sum the saturated compliances (Figure~\ref{fig:threea}), which is the original method of~\citet{Shafiro97}.
We call it the pore-impact approach.
An alternative conjecture is to treat connected pores as one large inhomogeneity (at the mesoscopic scale) and calculate the saturated compliances once per set only (Figure~\ref{fig:threeb}). 
We call it the set-impact approach.
As shown in the previous section, both methods are equivalent if pores in a set are identical, namely
\begin{equation}\label{twomet}
\sum_{c=1}^{m}\phi_c\Delta H_{ijk\ell_c}=\phi^{(p)}\Delta H^{(p)}_{ijk\ell}\,.
\end{equation}
Since the pore size affects the volume fraction only---whereas saturated compliances remain the same---the equation above also holds if pore sizes vary in a set.
This can be seen if we rewrite the volume fraction of a set as $\sum_{c=1}^m\phi_c=m\overline{\phi_c}=:\phi^{(p)}$, where the bar denotes an average.
Nevertheless, except for the two aforementioned cases, equation~(\ref{twomet}) is not generally obeyed.

There are situations when pore-impact and set-impact approaches predict approximately equal fluid effects.
Such scenarios happen if both product approximation $\overline{\phi_c\Delta H}_{ijk\ell_c}\approx\overline{\phi_c}\,\,\overline{\Delta H_{ijk\ell_c}}$ and relation $\overline{\Delta H_{ijk\ell_c}}\approx\Delta H_{ijk\ell}^{(p)}$ are satisfied, namely,
\begin{equation}\label{twomet4}
\sum_{c=1}^{m}\phi_c\Delta H_{ijk\ell_c}=m\overline{\phi_c\Delta H_{ijk\ell_c}}\approx m\overline{\phi_c}\,\,\overline{\Delta H_{ijk\ell_c}}\approx\phi^{(p)}\Delta H^{(p)}_{ijk\ell}\,.
\end{equation}
The product approximation holds if at least one variable is almost constant~\citep{Backus62} or if both variables are random and independently distributed.
On the other hand, $\overline{\Delta H}_{ijk\ell_c}\approx\Delta H_{ijk\ell}^{(p)}$ holds if $\Delta \bm{H}_c\approx\rm{const}$ that also satisfies the product approximation. Therefore, approximation~(\ref{twomet4}) can hold in the case of slightly varying (from pore to pore) saturated compliances that correspond to almost identical shapes and orientations of the inhomogeneities~\citep{Kach18}.
\begin{figure}[!htbp]
\centering
\begin{subfigure}{.4\textwidth}
  \centering
\includegraphics[scale=0.44]{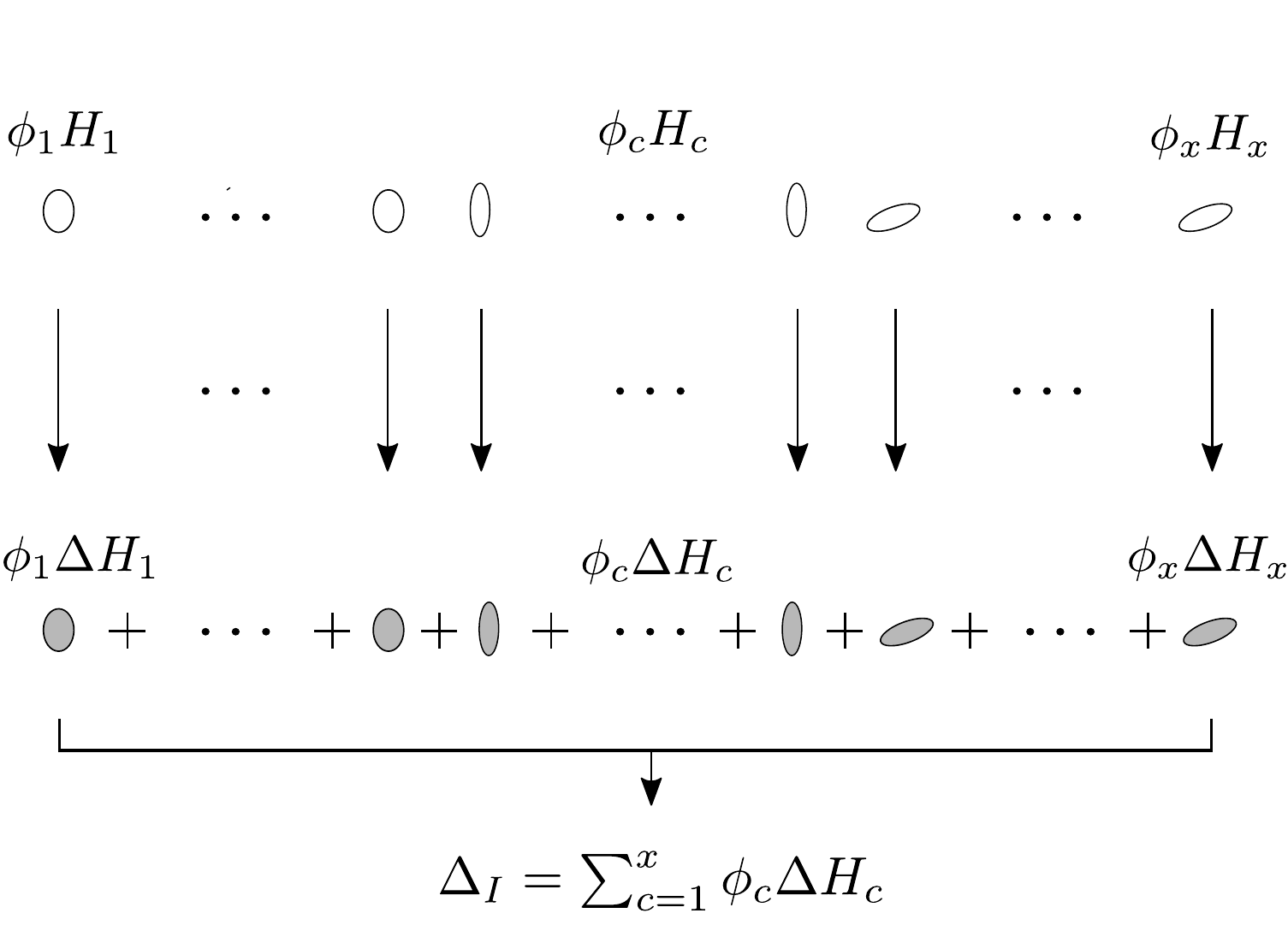}
\caption{\footnotesize{Pore-impact approach $(I)$}}
\label{fig:threea}
\end{subfigure}
\qquad\qquad
\begin{subfigure}{.4\textwidth}
  \centering
\includegraphics[scale=0.44]{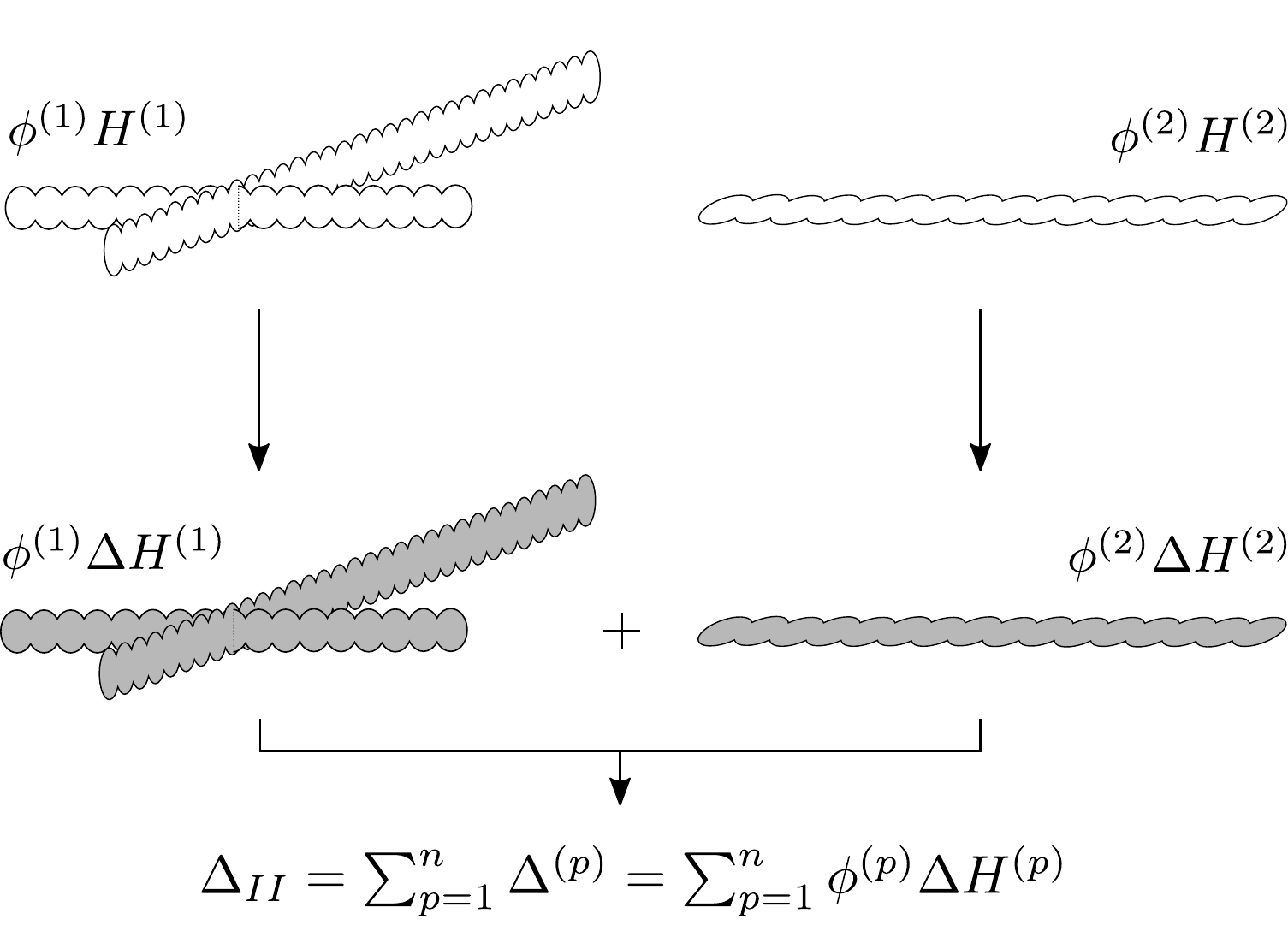}
\caption{\footnotesize{Set-impact approach $(II)$}}
\label{fig:threeb}
\end{subfigure}
\caption{\footnotesize{Illustration of the total fluid effect according to two different approaches, $\Delta_I$ and $\Delta_{II}$. In this example, two sets are embedded in the solid matrix. Grey colour symbolises fluids, whereas white space denotes dry pores. }}
\label{fig:three}
\end{figure}

Let us exemplify and discuss the pore-impact approach.
The effect of fluid is no longer related to the sets of pores so that the poroelastic parameters $S^{(p)}$ and $B_{ij}^{(p)}$ cannot be defined unless pores have identical shapes and orientations.
This can be seen if we compare the extended Biot description with the micromechanical approach, namely,
\begin{equation}
\Delta^{(p)}_{ijk\ell}=\sum_{c=1}^{m}\phi_c\Delta H_{ijk\ell_c}\,,
\end{equation}
which can be rewritten as
\begin{equation}\label{equivalence2}
-\frac{1}{9}S^{(p)}B_{ij}^{(p)}B_{k\ell}^{(p)}= -\phi_1\left(\frac{1}{K_{d_1}}+\frac{1}{K_f}-\frac{1}{K_0}\right)\,Q_{ij_1}Q_{k\ell_1}-\dots-\phi_m\left(\frac{1}{K_{d_{m}}}+\frac{1}{K_f}-\frac{1}{K_0}\right)\,Q_{ij_{m}}Q_{k\ell_{m}}\,.
\end{equation}
Storage and Skempton-like coefficients can be defined in terms of excess compliances only if $\Delta \bm{H}_c=\rm{const}$.
In view of the pore-impact approach, the geometry of the sets does not matter. 
Each pore is treated separately, and their fluid impacts are summed.
Therefore, using the pore-impact approach, one can dismiss consideration of the sets but rather consider a total fluid effect,
\begin{equation}\label{equivalence2b}
\Delta_{ijk\ell}=\sum_{c=1}^x\phi_c\Delta H_{ijk\ell_c}\,,
\end{equation}
where $x$ denotes the total number of pores that may allow fluid flow.
Hence, in general, the pore-impact description is inconvenient from the perspective of the extended poroelasticity, where sets---and their poroelastic parameters---are essential.
Further, according to expression~(\ref{equivalence2}), each pore in a set---due to various geometries ---is described by a different fluid polarisation tensor $Q_{ij_c}$ that implies different fluid pressure.
This is in contradiction to an extended Biot view that assumes constant pressure in the set that naturally results in constant pressure in each interconnected pore.
Another, more practical, downside of this description is that each pore needs to be considered separately, which is difficult and time-consuming to measure in the laboratory and is impossible to measure in the field.
Nevertheless, the pore-impact approach is in line with the effective methodology initiated by~\citet{Eshelby57}, continued by~\citet{OConnell77}, and standardised by~\citet{Shafiro97}, where saturation of each pore is considered separately. This approach can be used successfully to calculate the total fluid effect of the poroelastic medium.

To utilise the latter method, which we refer to as a ``set-impact'' approach, we should treat a pore set as one mesoscopic inhomogeneity.
To do so, first, we need to express $\sum_{c=1}^m\phi_cH_{ijk\ell_c}$ by a single volume fraction and one excess compliance tensor.
Hence, we utilise
\begin{equation}\label{def:Hp0}
\phi^{(p)}{{H}}_{ijk\ell}^{(p)}:=\sum_{c=1}^m\phi_cH_{ijk\ell_c}\,,
\end{equation}
where $\phi^{(p)}$ is the volume fraction of all pores in the set.
Note that we can define
\begin{equation}\label{def:Hp}
 H_{ijk\ell}^{(p)}:=\frac{\sum_{c=1}^m\phi_cH_{ijk\ell_c}}{\sum_{c=1}^m\phi_c}=\overline{H_{ijk\ell_c}}\,.
\end{equation}
Herein, the bar denotes the average weighted by volume fractions.
Second, we insert $H_{ijk\ell}^{(p)}$ inside~(\ref{DeltaH})--(\ref{Kc}) to get $\Delta H_{ijk\ell}^{(p)}$.
This way, the fluid effect in a set is obtained,
\begin{equation}\label{equivalence3}
\Delta^{(p)}_{ijk\ell}=\phi^{(p)}\Delta H_{ijk\ell}^{(p)}
\end{equation}
that can be rewritten as
\begin{equation}
-\frac{1}{9}S^{(p)}B_{ij}^{(p)}B_{k\ell}^{(p)}=-\phi^{(p)}\left(\frac{1}{K_d^{(p)}}+\frac{1}{K_f}-\frac{1}{K_0}\right)\,Q^{(p)}_{ij}Q^{(p)}_{k\ell}\,.
\end{equation}
It is clear that such a description is fully compatible with the extended poroelastic description. 
There is only one fluid polarisation tensor per set, which implies constant pressure, as expected.
Expressions~(\ref{multi_identical})--(\ref{def:skempton}) remain valid, also for various geometries.
We note that the treatment of the set as one mesoscopic inhomogeneity may be regarded as inconsistent from the perspective of the effective medium theory.
Therefore, at this stage, despite its obvious advantages, the set-impact approach should be treated as a conjecture only.

Let us review the achievements of Section~\ref{sec:21}.
First, we have pointed out a micromechanical analysis for the case of a single isolated pore.
We did not indicate the relationship between poroelastic Biot-like constants yet since we assumed that a single pore cannot allow fluid flow.
Instead, we used the derived expressions as a basis for the multi-pore scenarios.
We have shown that the effect of fluids present in the interconnected identical pores is equivalent to the impact of saturated compliances obtained from effective methods~(\ref{equivalence}). 
Poroelastic coefficients can be defined in terms of these compliances~(\ref{def:storage})--(\ref{def:skempton}). 
Similarly, the impact of fluids may be considered in the sets of pores having various geometries. 
Depending on the methodology involved, the fluid effect is either~(\ref{equivalence2b}) or~(\ref{equivalence3}); both approaches can be approximately equal in specific situations. 
Knowing the fluid impact, the strains of an undrained medium~(\ref{four}) can be described by both approaches.
Both effective methods can be utilised not only to take into account the strains caused by undrained connected pores but to consider isolated pores that form the skeleton, $S_{ijk\ell}$, as well. 
Nevertheless, there exists a significant drawback to one approach.
If the pore-impact method is utilised, then poroelastic parameters cannot be defined for a set generally. 
On the other hand, if we use the set-impact approach, then definitions~(\ref{def:storage})--(\ref{def:skempton}) are valid.
Therefore, in general, the fundamental equations~(\ref{one})--(\ref{two}) might be described using the latter micromechanical approach only.

Two essential related questions remain.
First, which micromechanical description of a fluid effect is more accurate, and when? 
Second, can we utilise the set-impact conjecture to describe the strains of partially-saturated medium and the resulting change of fluid content? 
To address these questions, we need to perform detailed laboratory experiments and compare the results with the theoretical predictions. 
However, this is beyond the scope of the current paper.
Herein, we describe both methodologies on specified microstructure and simulate numerically the fluid-effect discrepancies between the methods; which gives us indicative answers only.
\subsection{Micromechanical analysis: specified microstructure}
In the previous section, we have related the extended poroelasticity with micromechanics for unspecified excess compliances. 
These excess compliances can be obtained for various geometries such as cracks, spheres, needles, and others, which can be classified as ellipsoids~\citep{Kach18}.
Herein, we describe a few microstructures that are commonly considered in geophysics and tectonics and are interesting in the context of extended poroelasticity.
\subsubsection{Penny-shaped cracks: any orientations}
Let us consider a medium with one set of penny-shaped cracks having different orientations.
In such a case, we utilise
\begin{equation}\label{crack}
\sum_{c=1}^m\phi_cH_{ijk\ell_c}=\frac{1}{4}\left(\delta_{ik}\alpha_{jl}+\delta_{i\ell}\alpha_{jk}+\delta_{jk}\alpha_{i\ell}+\delta_{j\ell}\alpha_{ik}\right)+\beta_{ijk\ell}\,,
\end{equation}
where $\delta_{ij}$ is the Kronecker delta and
\begin{align}\label{alphabeta}
\alpha_{ij}&:=\sum_{c=1}^mZ_{T_c}n_{i_c}n_{j_c}\,,\\
\label{alphabetab}
\beta_{ijk\ell}&:=\sum_{c=1}^m\left(Z_{N_c}-Z_{T_c}\right)n_{i_c}n_{j_c}n_{k_c}n_{\ell_c}
\end{align}
stand for crack density tensors~\citep{Kachanov80}, where $n_i$ is the normal to the crack surface. 
The above tensors contain
\begin{equation}
Z_{T_c}:=\frac{32e_c}{3 E}\frac{\left(1-\nu^2\right)}{\left(2-\nu\right)}\,,\qquad
Z_{N_c}:=Z_{T_c}\left(1-\frac{\nu}{2}\right)\,,
\end{equation}
with
\begin{equation}
e_c=\frac{a_c^3}{V}\,,
\end{equation} 
where $V$ is the medium's volume, $a_c$ is the crack radius, $E$ and $\nu$ denote Young modulus and Poisson ratio of a solid phase, respectively.
If we consider the pore-impact approach, the fluid effect is
\begin{equation}
\sum_{c=1}^m\phi_c\Delta H_{ijk\ell_c}=-\sum_{c=1}^m\left(\frac{1}{1+\delta_c}\right)Z_{N_c}n_{i_c}n_{j_c}n_{k_c}n_{\ell_c}
\,,
\end{equation} 
where
\begin{equation}
\delta_c=K_{d_c}\left(\frac{1}{K_f}-\frac{1}{K_0}\right)=\frac{\phi_c}{Z_{N_c}}\left(\frac{1}{K_f}-\frac{1}{K_0}\right)=\frac{\gamma_c\pi E}{4\left(1-\nu^2\right)}\left(\frac{1}{K_f}-\frac{1}{K_0}\right)
\end{equation}
depends on the aspect ratio of each crack, $\gamma_c$.
If we consider the set-impact approach, we obtain $H^{(p)}_{ijk\ell}$ by dividing the right-side of expression~(\ref{crack}) by
\begin{equation}
\phi^{(p)}=\frac{4\pi}{3V}\sum_{c=1}^ma_c^3\gamma_c\,.
\end{equation}
Then, as prescribed earlier, we use $H^{(p)}_{ijk\ell}$ inside~(\ref{DeltaH})--(\ref{Kc}) to get $\Delta H_{ijk\ell}^{(p)}$ that enters~(\ref{equivalence3}).
\subsubsection{Penny-shaped cracks: orthogonal orientations}
Let us consider a simple example of perpendicular cracks.
Think of one interconnected set of cracks having identical shapes and sizes; meaning that $\gamma_c, a_c=\rm{const}$ that implies $Z_{N_c},Z_{T_c},\delta_c=\rm{const}$.
Assume that $m_1$ cracks have surface normals oriented towards the $x_1$-axis, $m_2$ towards the $x_2$-axis, and $m_3$ towards the $x_3$-axis, where the total number $m=m_1+m_2+m_3$.
This way, we can define $m_1Z_{N_c}=Z_{N1}$, $m_2Z_{N_c}=Z_{N2}$, $m_3Z_{N_c}=Z_{N3}$.
The fluid effect according to the pore-impact $(I)$ description is (we show the non-zero $3\times3$ minor only),
\begin{equation}
\Delta_{I}=
-\frac{1}{S_c}\left[
\begin{array}{ccc}
\frac{1}{m_1}Z^2_{N1} & 0 & 0 \\
0& \frac{1}{m_2}Z^2_{N2} & 0 \\
0 & 0 & \frac{1}{m_3}Z^2_{N3} \\
\end{array}
\right]
=
-\frac{1}{1+\delta}\left[
\begin{array}{ccc}
Z_{N1} & 0 & 0 \\
0& Z_{N2} & 0 \\
0 & 0 & Z_{N3} \\
\end{array}
\right]\,,
\end{equation}
where 
\begin{equation}
S_c=\phi_c\left(\frac{1}{K_{d_c}}+\frac{1}{K_f}-\frac{1}{K_0}\right)\,,\qquad \frac{1}{K_{d_c}}=\frac{Z_{N_c}}{\phi_c}\,.
\end{equation}
The pore-impact approach disregards the connections between cracks, hence, the storage-like coefficient $S_c$ is obtained for each pore. 
Note that in each subset, $m_i$, cracks are identical, which is an exceptional situation.
Therefore, instead of treating each crack individually, we can rewrite the above expressions in terms of poroelastic constants that correspond to three subsets (isolated or not!) of identical cracks. 
In other words,
\begin{equation}
\Delta_I=-\frac{1}{9S^{(1)}}
\left[
\begin{array}{ccc}
\left(S^{(1)}B_{11}^{(1)}\right)^2 & 0 & 0 \\
0&0&0 \\
0 & 0& 0 \\
\end{array}
\right]\,
-
\frac{1}{9S^{(2)}}
\left[
\begin{array}{ccc}
0 & 0 & 0 \\
0&\left(S^{(2)}B_{22}^{(2)}\right)^2&0 \\
0 & 0& 0 \\
\end{array}
\right]\,
-
\frac{1}{9S^{(3)}}
\left[
\begin{array}{ccc}
0 & 0 & 0 \\
0&0&0 \\
0 & 0& \left(S^{(3)}B_{33}^{(3)}\right)^2 \\
\end{array}
\right]\,,
\end{equation}
where
\begin{equation}
S^{(i)}=m_i\phi_c\left(\frac{1}{K_{d_c}}+\frac{1}{K_f}-\frac{1}{K_0}\right)\,, \qquad B^{(i)}_{ii}=\frac{3Z_{Ni}}{S^{(i)}}\,.
\end{equation}

On the other hand, the set-impact $(II)$ approach indicates
\begin{equation}
\Delta_{II}=
-\frac{1}{S}
\left[
\begin{array}{ccc}
Z_{N1}^2 & Z_{N1}Z_{N2} & Z_{N1}Z_{N3} \\
Z_{N1}Z_{N2}& Z_{N2}^2 & Z_{N2}Z_{N3} \\
Z_{N1}Z_{N3} & Z_{N2}Z_{N3} & Z_{N3}^2 \\
\end{array}
\right]\,
=
-\frac{1}{9S}
\left[
\begin{array}{ccc}
S^2B_{11}^2 & S^2B_{11}B_{22} & S^2B_{11}B_{33} \\
S^2B_{11}B_{22}& S^2B_{22}^2 & S^2B_{22}B_{33} \\
S^2B_{11}B_{33} & S^2B_{22}B_{33} & S^2B_{33}^2 \\
\end{array}
\right]\,,
\end{equation}
where
\begin{equation}
S=m\phi_c\left(\frac{1}{K_d^{(p)}}+\frac{1}{K_f}-\frac{1}{K_0}\right)\,, \qquad \frac{1}{K_d^{(p)}} 
=\frac{1}{K_{d_c}}\,,
\qquad
B_{ii}=\frac{3Z_{Ni}}{S}\,.
\end{equation}
We note that the set-impact approach leads to one storage coefficient and one Skempton tensor only; cracks are described as connected.
In our example, the connections between each $m_i$ subset are expressed as the non-zero off-diagonal terms of $\Delta_{II}$.
The aforementioned terms are absent in the former methodology, where cracks are treated separately.
Importantly, if each subset constitutes a detached set, then $\Delta_I$ remains the same (connections do not matter) but $\Delta_{II}$ reduces to $\Delta_I$ (due to sets with identical cracks).
A significant influence of connections between subsets on the set-impact ($II$) description is clear.
\subsubsection{Penny-shaped cracks: TI orientations}\label{sec:CTI}
As mentioned in the Introduction, the effective transverse isotropy can be obtained by distributing TI-oriented penny-shaped cracks in the isotropic solid phase~\citep{Sayers95}.
Following the rock physics nomenclature, we refer to such a particular case briefly as cylindrical transverse isotropy (CTI). 
Commonly, the microstructure that leads to CTI is described as vertical cracks that are not aligned but are randomly distributed around the symmetry axis. 
Its excess compliances are expressed by a general formula for cracks~(\ref{crack}), where additionally $\alpha_{11}=\alpha_{22}$, $\beta_{1111}=\beta_{2222}$, and $\beta_{1122}=\beta_{1111}/3$. 
Also, note that $\alpha_{ij}$ and $\beta_{ijk\ell}$ are symmetric with respect to all rearrangements of indices.

Due to the random orientation of cracks viewed in horizontal plane, the link between the geometry of particular pores and crack density tensors is unclear. 
In other words, it is not easy to obtain values of crack density tensors---that satisfy the CTI relations---by inserting each particular crack radius and surface normals.
Therefore, the microstructure that induces CTI is treated holistically---as one large group of vertical cracks, where the orientation of a specific crack is unknown.  
Optionally, the horizontal cracking can be considered---coefficients $\alpha_{33}$ and $\beta_{3333}$ increase, whereas the CTI relations remain obeyed.  
If certain cracks are connected, then they allow fluid flow, and the poroelasticity theory can be used.
In the case that all cracks are interconnected (single set), the original Biot theory is applicable (one Skempton tensor and single storage coefficient).
The situation of an isolated set of horizontal inhomogeneities (two sets in total) can be furnished by poroelastic extension (two Skempton-like tensors and two storage coefficients). 
Both scenarios are depicted in Figures~\ref{fig:foura}--\ref{fig:fourb}.

What if the medium is CTI but pore pressure in vertical cracks is not constant? In other words, what if---although TI oriented---the vertical pores are forming a few or dozens of isolated sets?  
To furnish such a situation, we propose to specify the geometry of each pore and use these geometries to recreate the crack density tensors that satisfy CTI relations.
To do so, we assume $n$ sets of aligned vertical cracks that are embedded in the isotropic solid phase and are equally distributed around the symmetry axis. 
In other words, each vertical set is isolated from the other by an angle equal to $\theta=\pi/n$ radians.
According to Herman's theory~\citep{Herman45}, a fourth-rank tensor is transversely isotropic when it is invariant to at least five-fold rotation about the symmetry axis.
Thus, the smallest number of symmetry planes distributed equally around the symmetry axis that would satisfy Herman's requirement is three.  
By analogy, if a pore set corresponds to the symmetry plane, equally distributed sets could also lead to TI symmetry.
In fact, as we prove in Appendix~\ref{sec:ap1}, $n\geq3$ sets are sufficient to obtain TI symmetry induced by aligned cracks---this agrees with Herman's theory.
However, each set must have an equal number of cracks of the same size; that is a necessary condition to be satisfied. 
On the other hand, two orthogonal symmetry planes lead to orthotropy. 
As indicated by~\citet{Sch97}, $n=2$ embedded sets induce such symmetry that supports our crack-set symmetry-plane analogy.
In the case of set-induced orthotropy, there is no requirement for an equal crack number or size in each set.
\begin{figure}[!htbp]
\centering
\begin{subfigure}{.4\textwidth}
  \centering
\includegraphics[scale=0.33]{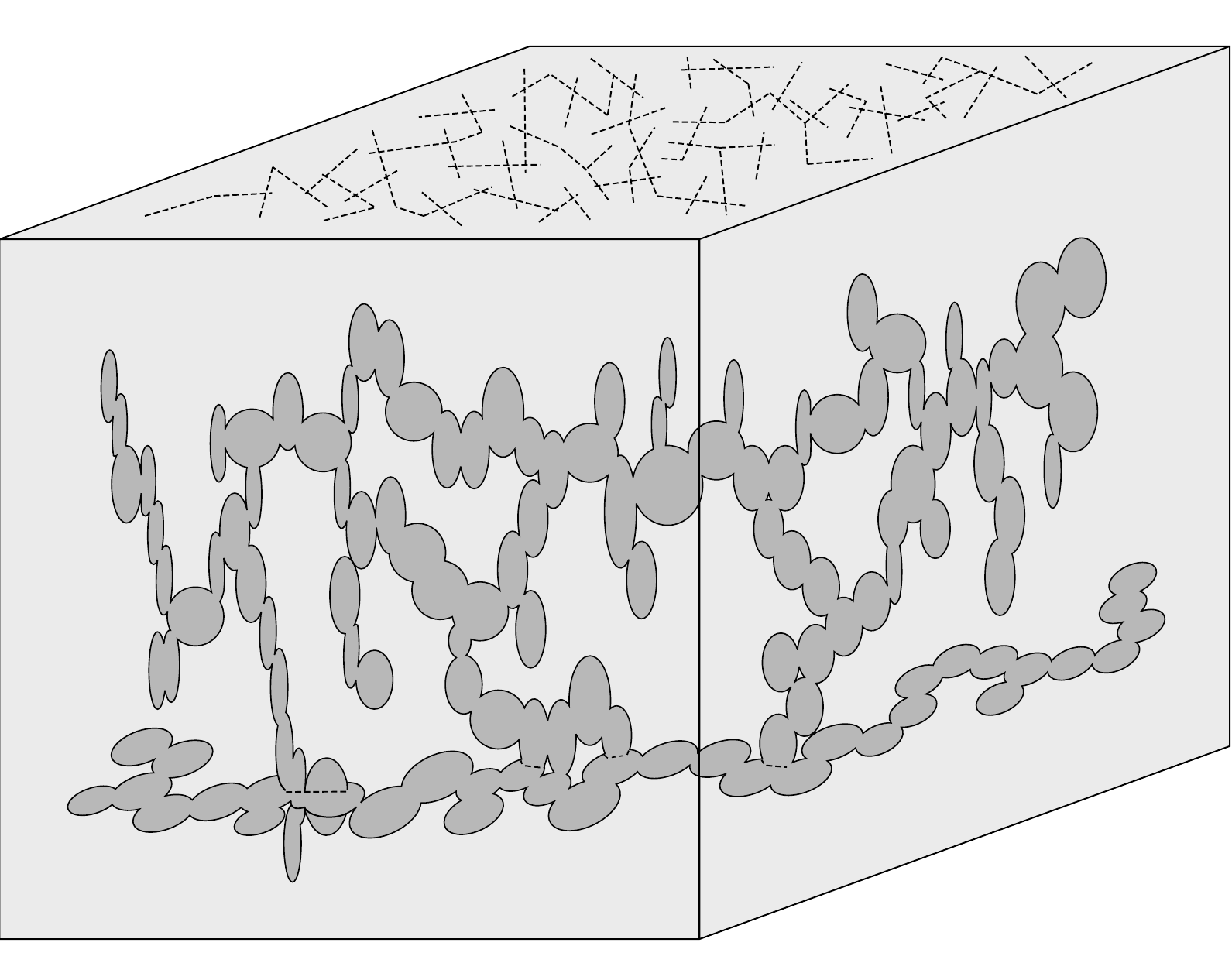}
\caption{\footnotesize{Classic CTI: connected case}}
\label{fig:foura}
\end{subfigure}
\qquad\qquad
\begin{subfigure}{.4\textwidth}
  \centering
\includegraphics[scale=0.33]{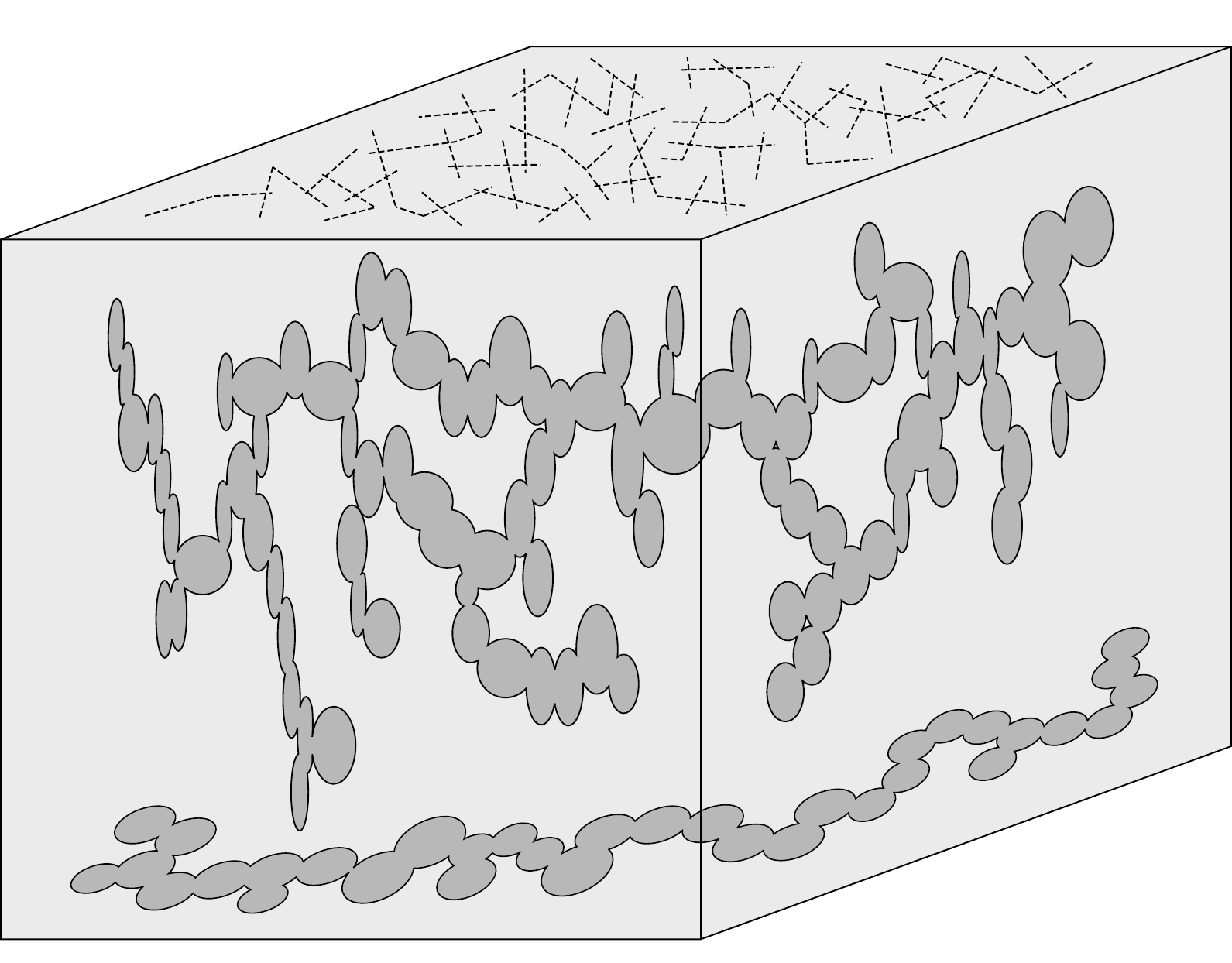}
\caption{\footnotesize{Classic CTI: isolated case}}
\label{fig:fourb}
\end{subfigure}
\qquad\qquad
\begin{subfigure}{.4\textwidth}
  \centering
\includegraphics[scale=0.33]{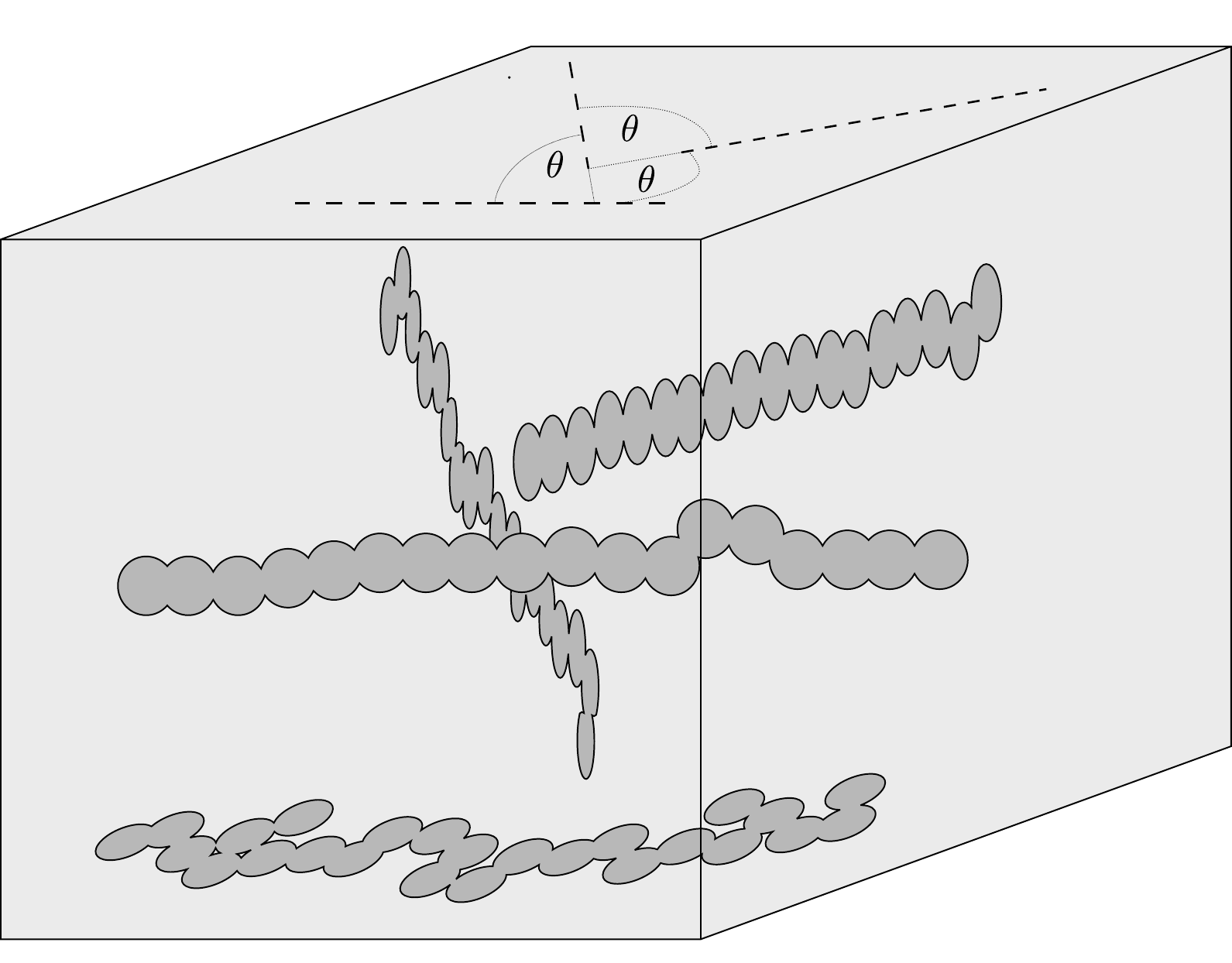}
\caption{\footnotesize{Set-induced CTI: isolated case}}
\label{fig:fourc}
\end{subfigure}
\qquad\qquad
\begin{subfigure}{.4\textwidth}
  \centering
\includegraphics[scale=0.33]{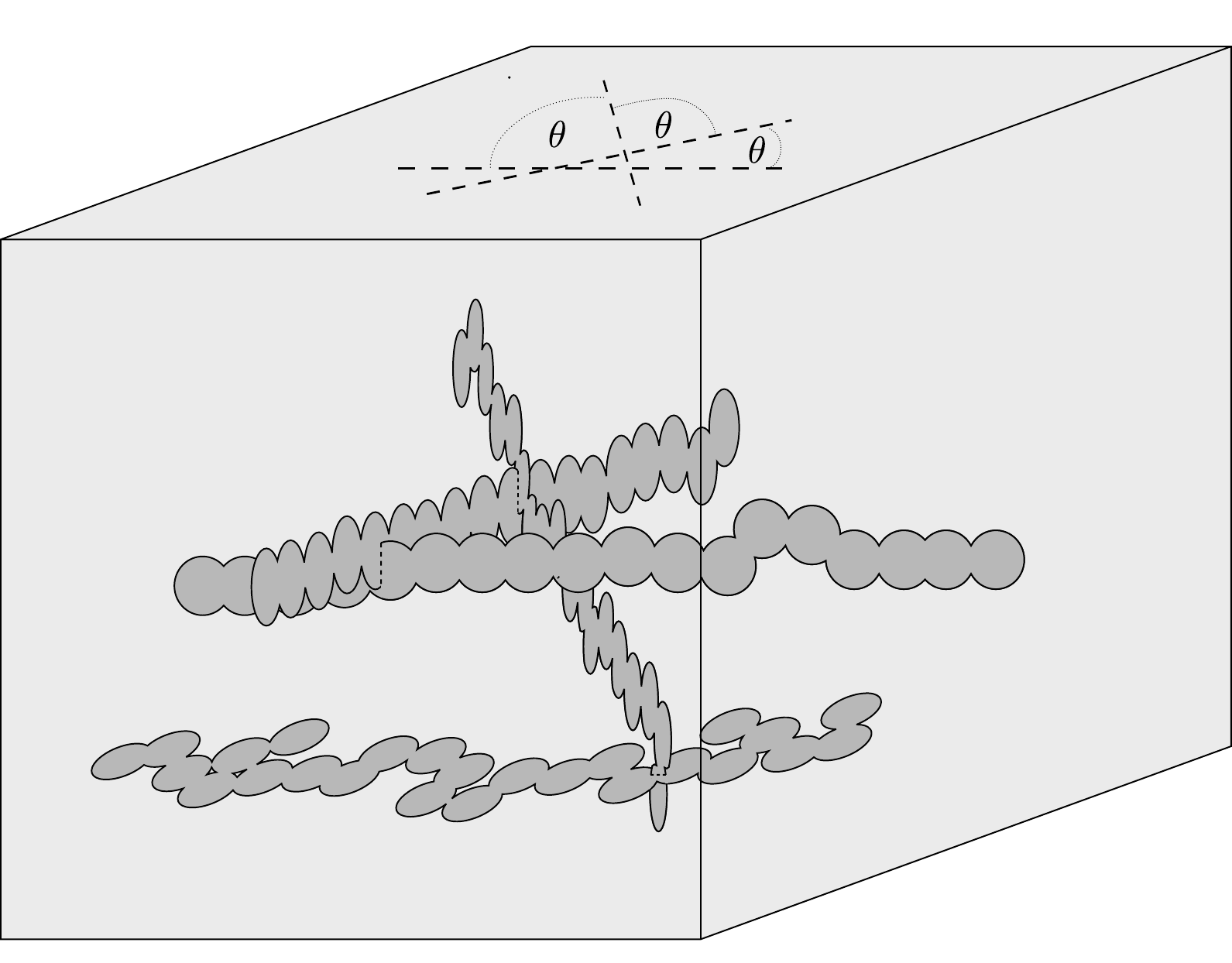}
\caption{\footnotesize{Set-induced CTI: connected case}}
\label{fig:fourd}
\end{subfigure}
\caption{\footnotesize{Some CTI scenarios of the poroelastic medium are illustrated. In each case, we include horizontal cracking that is optional; its absence does not affect the effective symmetry. First two figures represent classic cases of random transverse-isotropic orientations of vertical cracks. $\bm{a)}$ depicts interconnected vertical and horizontal cracks; they form one set only. $\bm{b)}$ presents an isolation of vertical and horizontal inhomogeneities; they form two distinct sets. 
The last two figures represent vertical cracks that are not random but aligned in $n=3$ directions.
The alignments are isolated horizontally by $\theta$ and are responsible for the CTI inducement. $\bm{c)}$ illustrates the scenario where all sets are detached ($4$ sets in total), whereas $\bm{d)}$ depicts the interconnected case (one set only). In the context of pore-impact approach or dry excess compliances, case $\bm{a)}$ and $\bm{c)}$ is equivalent to $\bm{b)}$ and $\bm{d)}$, respectively. The discrepancies arise if the set-impact approach is used, where connections do matter.}}
\label{fig:four}
\end{figure}

Let us provide an example to illustrate the general proof from Appendix~\ref{sec:ap1}.
Consider three sets of vertical and aligned cracks. Sets are isolated by $\theta=\pi/3$ and their surface normals $\bm{n}^{(p)}$ are
\begin{equation}
\bm{n}^{(1)}=[\cos (\theta),\,\sin (\theta),\,0]\,,\qquad \bm{n}^{(2)}=[\cos (2\theta),\,\sin (2\theta),\,0]\,,\qquad \bm{n}^{(3)}=[\cos (3\theta),\,\sin (3\theta),\,0]\,.
\end{equation}
Assume that sets are identical, meaning that the number of cracks and their shapes are the same in each set. 
Thus, we can state that
\begin{equation}
\sum_{c=1}^mZ_{N_c}=Z_N^{(p)}=Z_N\,
\end{equation}
and analogous description holds for $Z_{T_c}$.
Crack density tensors of the effective medium are
\begin{align}\label{denstens}
\alpha_{ij}&=\sum_{p=1}^n\alpha_{ij}^{(p)}=Z_T\sum_{p=1}^nn_i^{(p)}n_j^{(p)}\,,\\
\label{denstensb}
\beta_{ijk\ell}&=\sum_{p=1}^n\beta_{ijk\ell}^{(p)}=(Z_N-Z_T)\sum_{p=1}^nn_i^{(p)}n_j^{(p)}n_k^{(p)}n_\ell^{(p)}\,.
\end{align}
We get
\begin{align}
\alpha_{11}&=Z_T\left(\frac{1}{4}+\frac{1}{4}+1\right)=\frac{3Z_T}{2}\,,\\
\alpha_{22}&=Z_T\left(\frac{3}{4}+\frac{3}{4}+0\right)=\alpha_{11}\,,\\
\beta_{1111}&=(Z_N-Z_T)\left(\frac{1}{16}+\frac{1}{16}+1\right)=(Z_N-Z_T)\frac{9}{8}\,,\\
\beta_{2222}&=(Z_N-Z_T)\left(\frac{9}{16}+\frac{9}{16}+0\right)=\beta_{1111}\,,\\
\beta_{1122}&=(Z_N-Z_T)\left(\frac{3}{16}+\frac{3}{16}+0\right)=\frac{1}{3}\beta_{1111}\,,
\end{align}
which indicate CTI. 
Note that such symmetry would not appear if sets were not identical.
Also, it is apparent that $\alpha_{33}=\beta_{3333}=\beta_{1133}=0$ and
\begin{equation}\label{relationship}
\alpha^*_{11}=\frac{4}{3}\beta^*_{1111}\,,
\end{equation}
where
\begin{align}\label{alphabetaa}
\alpha_{ij}^*&:=\sum_{c=1}^mn_{i_c}n_{j_c}\,,\\
\beta_{ijk\ell}^*&:=\sum_{c=1}^mn_{i_c}n_{j_c}n_{k_c}n_{\ell_c}
\end{align}
can be defined as crack fabric tensors~\citep[sensu][]{Oda86}.
Relationship~(\ref{relationship}) is characteristic of any set-induced CTI, which can be proven as follows. 
For any $n\geq3$ vertical sets, relation $\beta_{1133}=\beta^*_{1133}=0$ must be true.
As a consequence, 
\begin{equation}\label{relationship2}
\alpha^*_{11}=\sum_{k=1}^3\beta^*_{11kk}=\beta^*_{1111}+\frac{1}{3}\beta^*_{1111}+0=\frac{4}{3}\beta^*_{1111}\,,
\end{equation}
as required.

Inserting~(\ref{denstens})--(\ref{denstensb}) into expression~(\ref{crack}), we get total excess compliances for a dry effective medium. 
Using condensed $6\times6$ matrix notation used by~\citet{Sch95} or~\citet{Kach18}, we obtain  
\begin{equation}
\sum_{p=1}^n\phi^{(p)}H_{ijk\ell}^{(p)}=
\frac{n}{8}\left[
\begin{array}{cccccc}
3Z_N+Z_T & Z_N-Z_T& 0 & 0 & 0 & 0 \\
Z_N-Z_T & 3Z_N+Z_T& 0 & 0 & 0 & 0 \\
0 & 0& 0 & 0 & 0 & 0 \\
0 & 0& 0 & 4Z_T & 0 & 0 \\
0 & 0& 0 & 0 & 4Z_T & 0 \\
0 & 0& 0 & 0 & 0 & Z_N+7Z_T \\
\end{array}
\right]\,,
\end{equation}
which is universal for $n\geq3$.

Interestingly, CTI can be also obtained if---additionally to the $n$ vertical sets---we insert a horizontal set. 
For ease of representation, assume that these horizontal cracks are identical, namely, $\sum_{c=1}^mZ_{N_c}=Z_{N_H}$.
Nevertheless, we allow them to differ from vertical ones.
We get additional excess compliances,
\begin{equation}
\phi^{(\rm{hor})}H_{ijk\ell}^{(\rm{hor})}=
\left[
\begin{array}{cccccc}
0 & 0 & 0 & 0 & 0 & 0 \\
0& 0& 0 & 0 & 0 & 0 \\
0 & 0& Z_{N_H} & 0 & 0 & 0 \\
0 & 0& 0 & Z_{T_H} & 0 & 0 \\
0 & 0& 0 & 0 & Z_{T_H} & 0 \\
0 & 0& 0 & 0 & 0 & 0 \\
\end{array}
\right]\,.
\end{equation}
The necessary TI relations are still satisfied. Additionally, condition $\beta_{1133}=0$ must be obeyed.

The procedure of obtaining the excess compliances responsible for the effect of fluid, $\Delta H_{ijk\ell}$, is analogous to the case of orthogonal cracking, discussed in the previous section. If $n$ sets are isolated, then pore-impact and set-impact approaches are equivalent due to identical cracks in each set. If sets are not isolated, then they are treated as interconnected subsets that form a single porosity; both micromechanical approaches must differ. Importantly, the possible connections between subsets do not influence the above dry-case derivations. Isolated and non-isolated scenarios are illustrated in Figures~\ref{fig:fourc}--\ref{fig:fourd}.



\subsubsection{Spheres}\label{sec:sphere}
Let us now consider pores that have the shape of a sphere. Naturally, such a shape leads to isotropic dry excess compliance (the orientation of a pore does not influence $H_{ijk\ell}$).
In turn, 
\begin{equation}
{\bm{H}}_{_c}={\rm{const}}\implies \Delta {\bm{H}}_c={\rm{const}}\implies \Delta H_{ijk\ell_c}=\Delta H_{ijk\ell}^{(p)}\,.
\end{equation}
Even though orientations and shapes are identical, the sizes can vary, $\phi_c\neq{\rm{const}}$.
The fluid effect of pores is
\begin{equation}
\Delta_{I_{ijk\ell}}=\sum_{c=1}^x\phi_c \Delta H_{ijk\ell_c}=\Delta H_{ijk\ell}^{(p)}\sum_{c=1}^x\phi_c=\sum_{p=1}^{n<x}\phi^{(p)}\Delta H_{ijk\ell}^{(p)}=\Delta_{{II}_{ijk\ell}}\,.
\end{equation}
In other words, for any set of multiple spheres, both descriptions of a fluid impact are equivalent. 
If sizes and numbers of spheres forming distinct sets are different, then $\phi^{(p)}\neq{\rm{const}}\implies S^{(p)}\neq{\rm{const}}$. 
On the other hand, Skempton-like coefficients are identical for each set; they are explicitly impacted by ${\bm{H}}^{(p)}={\rm{const}}$ only.
\subsubsection{Non-spheroidal shapes}\label{sec:supersphere}
So far, we have described a few microstructures of spheroidal pores or cracks in the context of fluid effect and extended poroelasticity. However, as mentioned earlier, excess compliances can also be obtained for ellipsoids; therefore, the linkage between poroelasticity and effective methods is also allowable in the case of these more general shapes. 
Furthermore, such a linkage may exist for non-ellipsoidal shapes. As discussed by~\citet{Grechka06c}, excess compliances can be obtained for various---not necessarily penny shaped---cracks. In the case of non-flat pores, the approximations for non-ellipsoids are summarised in~\citet[Chapter 4.3,][]{Kach18}. Herein, we invoke the case of superspheres that may represent concave or convex pores. A surface of a supersphere of unit radius is described by $x_1^{(2k)}+x_2^{(2k)}+x_3^{(2k)}=1$, where parameter $k$ is a concavity factor. For $k<0.5$, the shape is concave, for $k>0.5$, it is convex, and for $k=1$ it is sphere. The excess compliances are approximately zero for $k<0.2$. If $k\in[0.2,\,1]$, we get
\begin{equation*}
H_{ijk\ell_c}\approx\frac{5k-1}{4}\frac{V_s}{V_1(k)}H_{ijk\ell_c}^s\,,
\end{equation*}
where $H^s_{ijk\ell_c}$ are the excess compliances of a sphere. $V_1$ and $V_s$ denote the volume of a supersphere~\citep[Expression 4.3.14,][]{Kach18} and a unit sphere, respectively. Subsequently, saturated compliances and fluid effects can be obtained.
For $k<1$, the excess compliance tensor is approximately isotropic, which means that both micromechanical approaches (pore-impact and set-impact) are approximately equal.
Similarly to the aspect ratio that is essential for ellipsoids, the concavity factor is crucial for superspheres---it has a strong effect on the pore contribution to the effective elasticity.
As shown by~\citet{Chen18}, the effect of both can be combined in the case of oblate or prolate superspheres. Note that the excess compliance tensor of an oblate or prolate supersphere is not approximately isotropic due to the effect of $\gamma$; hence, the two micromechanical approaches will generate different results.
\section{Numerical simulations}
In this section, we first show the differences between original and extended Biot theories applied to subsets of pores being connected or isolated, respectively. We utilise a scenario of identical pore shapes so that the choice of a micromechanical description does not influence the isolated case. Second, we focus on the discrepancies between pore-impact and set-impact approaches. To simplify the problem, we assume a single porosity only.  
In our simulations, we utilise the properties of Berea sandstone reported by~\citet{Beeler00} and~\citet{Wong17}, where for the solid phase $E=87\,\rm{GPa}$, $\nu=0.11$, and fluid compressibility $1/K_f=0.45\,\rm{GPa}^{-1}$. Additionally, in the second part, we compare the Berea sandstone with a much different solid phase, $E=125\,\rm{GPa}$ and $\nu=0.25$, typical for basalts~\citep{Ji10}. Results of our numerical experiments should not be treated as conclusive but indicative of certain repeatable phenomena.
\subsection{Original versus extended poroelasticity}\label{sec:numbiot}
Assume three vertical subsets of $m=100$ cracks each, having identical sizes and shapes, where $e_c=1/3m$ and $\gamma_c=0.01$. Let these subsets be equally distributed around the vertical axis. This way, CTI symmetry is induced. We assume a slight misalignment of subsets with the coordinate axes; crack azimuths with respect to $x_1$-axis are $\varphi^{(1)}=10^{\circ}$, $\varphi^{(2)}=130^{\circ}$, and $\varphi^{(3)}=250^{\circ}$. Knowing the stresses, the description of a porous medium behaviour can differ significantly depending on the connections between the subsets. 

\begin{figure}[!htbp]
\centering
\begin{subfigure}{.4\textwidth}
  \centering
\includegraphics[scale=0.44]{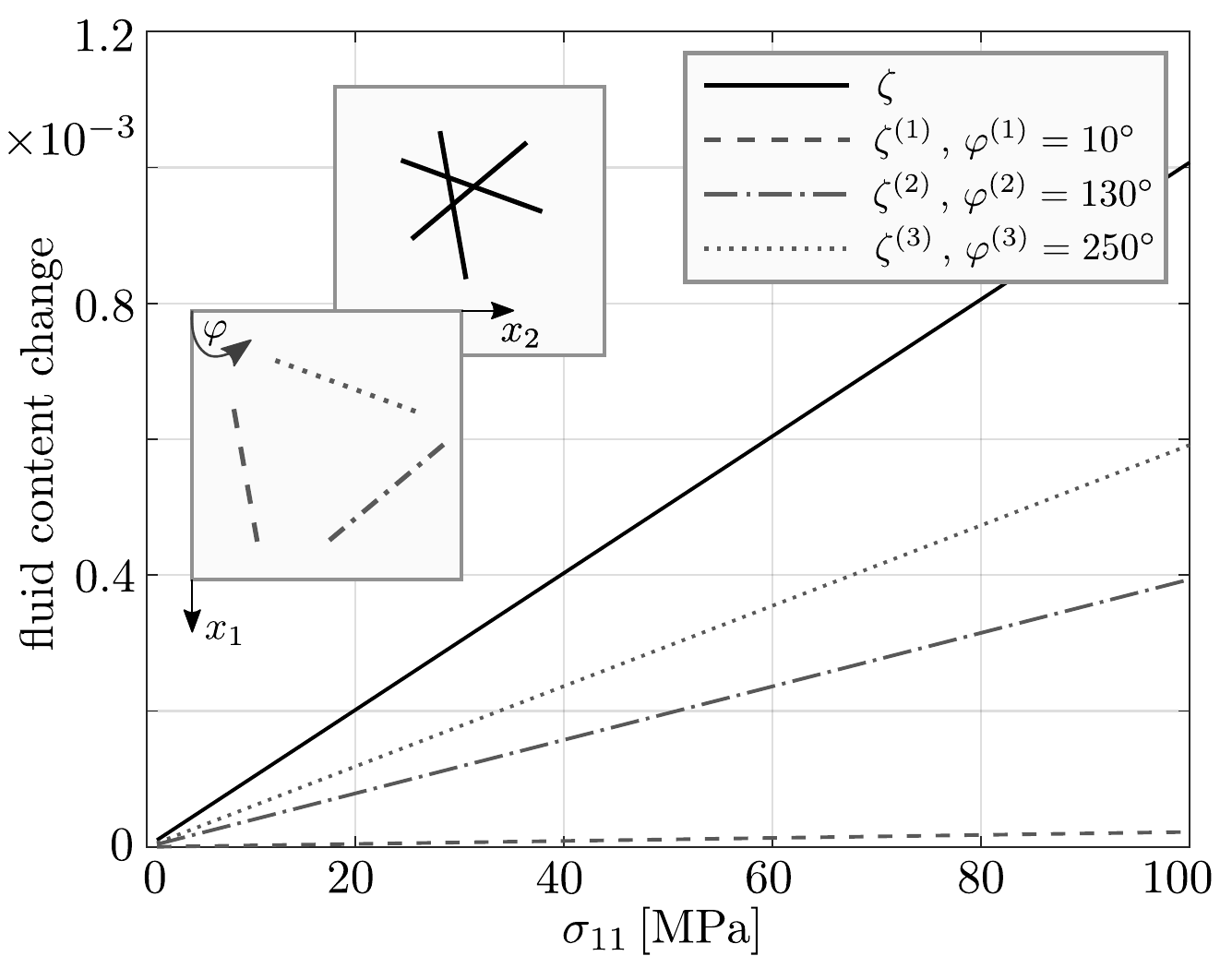}
\caption{\footnotesize{stress vs fluid content change}}
\label{fig6a}
\end{subfigure}
\qquad\qquad
\begin{subfigure}{.4\textwidth}
  \centering
\includegraphics[scale=0.44]{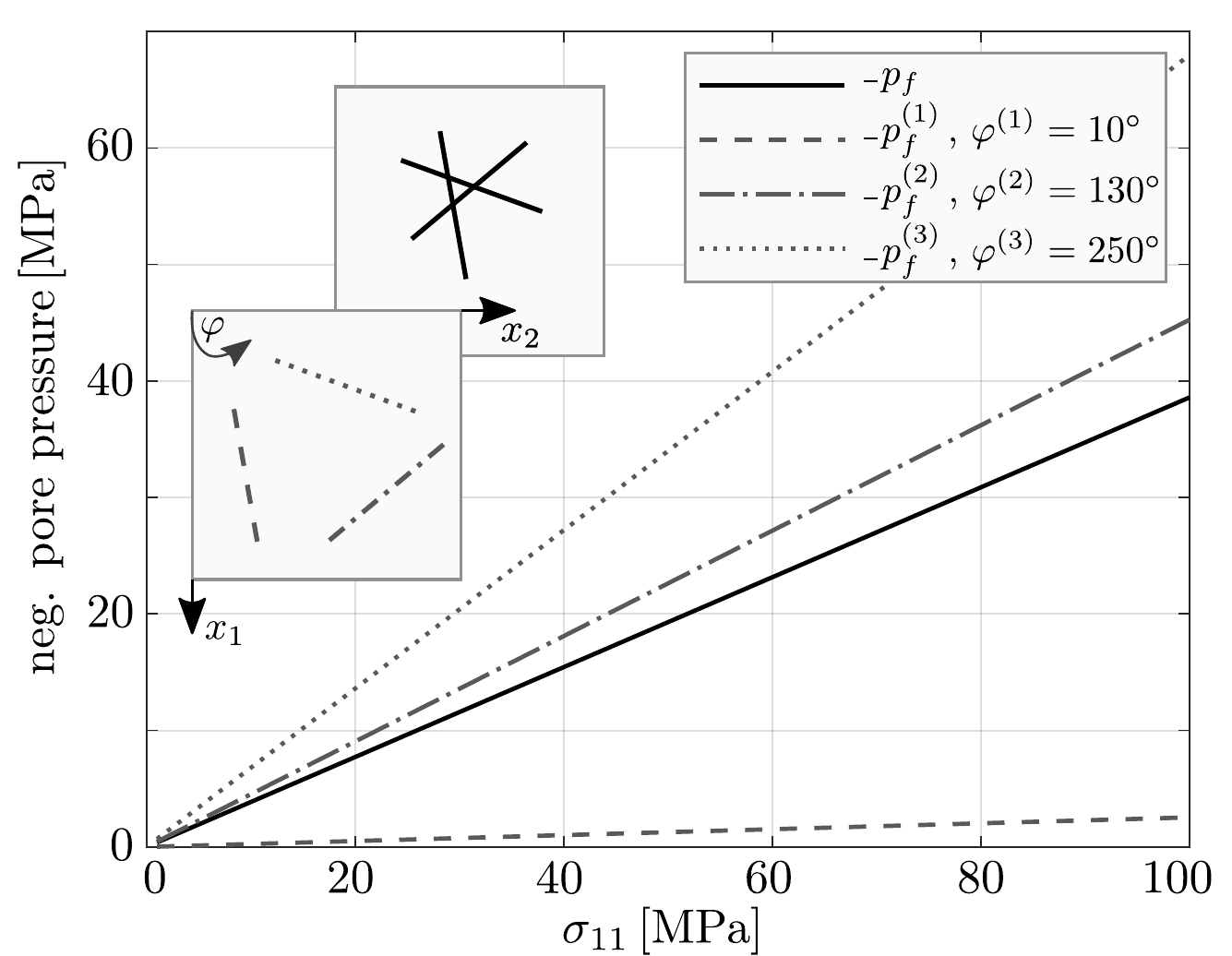}
\caption{\footnotesize{stress vs pressure}}
\label{fig6b}
\end{subfigure}
\caption{\footnotesize{Discrepancy between original and extended poroelasticity based on three subsets that induce CTI symmetry. If subsets are connected, then original Biot theory can be used (solid line). Isolated subsets require extended theory (non-solid lines). Schematic 2d views of the two cases are shown. Dashed line stands for vertical set with azimuth $\varphi^{(1)}=10^\circ$, dashed-dotted line indicates $\varphi^{(2)}=130^\circ$, and dotted line $\varphi^{(3)}=250^\circ$. We consider set-drained \textbf{(a)} and undrained \textbf{(b)} conditions.}}
\label{fig:six}
\end{figure}

If all subsets are connected, they form a single porosity and the original Biot theory is used. Using the set-impact approach, we obtain a single storage coefficient $S\approx0.0261\,\rm{GPa}^{-1}$ and non-zero Skempton components $B_{11}=B_{22}\approx1.16\,, B_{33}\approx0.01$. As a consequence, we can describe fluid content change $\zeta$ (drained case) or pore pressure $p_f$ (undrained case) as a function of stress, as shown by solid lines in Figures~\ref{fig6a}--\ref{fig6b}.

If all subsets are isolated, they form three distinct porosity clusters and the extended theory must be used. We obtain $S^{(1)}=S^{(2)}=S^{(3)}\approx0.0087\,\rm{GPa}^{-1}$ and
\begin{equation}
\bm{B}^{(1)}\approx
\left[
\begin{array}{ccc}
1.36 & -2.27 & 0 \\
-2.27& 0.96 & 0 \\
0 & 0 & 0.01 \\
\end{array}
\right]\,,\quad
\bm{B}^{(2)}\approx
\left[
\begin{array}{ccc}
2.04 & 1.48 & 0 \\
1.48 & 0.28 & 0 \\
0 & 0 & 0.01 \\
\end{array}
\right]\,,\quad
\bm{B}^{(3)}\approx
\left[
\begin{array}{ccc}
0.08 & 0.79 & 0 \\
0.79 & 2.24 & 0 \\
0 & 0 & 0.01 \\
\end{array}
\right]\,.
\end{equation}
As expected, in Figures~\ref{fig6a}--\ref{fig6b}, there are three (non-solid) lines depicting the relations of fluid content changes or pore pressure with the $\sigma_{11}$ stress component.
Due to the low aspect ratio, the values of these parameters depend strongly on the orientation of the subsets (the orientation would not matter in the case of spheres).  
The largest fluid content change and pressure occur if the long axes of the cracks are almost perpendicular to the stress direction. By analogy, the smallest values correspond to the case of the crack long axes being almost parallel to the stress. Therefore, in Figures~\ref{fig6a}--\ref{fig6b}, the subset $\varphi^{(3)}=250^{\circ}$ presents the largest values of $\zeta^{(p)}$ and $p_f^{(p)}$. Here, the crack axes are misaligned with $x_2$ by only $20^{\circ}$. The subset $\varphi^{(1)}=10^{\circ}$ has the lowest values of $\zeta^{(p)}$ and $p_f^{(p)}$. Here, the crack axes are misaligned with $x_1$ by only $10^{\circ}$. The remaining subset $\varphi^{(2)}=130^{\circ}$ presents moderate fluid content change and pore pressure due to relatively large misalignments with both $x_1$ and $x_2$ axes. Note that in the undrained conditions
\begin{equation}\label{storages}
S=\sum^n_{p=1}S^{(p)}
\end{equation}
and
\begin{equation}\label{skemptons}
B_{ij}=\overline{B_{ij}^{(p)}}\, \implies \, p_f=\overline{p_f^{(p)}}\,.
\end{equation}
As shown in Appendix~\ref{sec:ap}, equation~(\ref{storages}) holds in any circumstances. On the other hand,~(\ref{skemptons}) is obeyed due to $K_d^{(p)}$ being constant through sets that is usually not true. 
In our case, constant bulk moduli of pore sets are a consequence of identical pore shapes ($\gamma_c=\rm{const}$)---the requirement of induced CTI. 
Additionally, in drained conditions of this specific scenario, the constant volume fraction of each pore-set leads to $S^{(p)}=\rm{const}$ that implies $\zeta=\sum_{p=1}^n\zeta^{(p)}$ (see Appendix~\ref{sec:ap}); as depicted in Figure~\ref{fig6a}.
\subsection{Pore-impact versus set-impact approach}
To check the discrepancies between the two micromechanical descriptions, we performed multiple simulations for various numbers of $m$ pores forming a single set, where $m=3\times y$, $y\in[1, 1000]$, $y\in\mathbb{N}$. For each $m$, we quantify the aforementioned discrepancy as a relative difference ($R$) between pore-impact and set-impact approaches,
\begin{equation}
R^{f}_I:=\frac{||\Delta_I-\Delta_{II}||}{||\Delta_I||}\times100\%
\,,
\end{equation}
where $||\cdot||$ denotes a Frobenius norm and superscript $f$ stands for the effect of fluid.
We consider different shapes (aspect ratio $\gamma_c$), orientations, and sizes of pores (density $e_c=a^3_c/V$). Each characteristic can be identical ($i$), slightly varying ($s$), random ($r$), or can form a certain pattern ($p$) in the pore set. 
For instance, the case of pores having identical shapes, random orientations, and non-random sizes, would be denoted by $irp$, where the first letter always refers to shape, the second to orientation, and the third to pore size. Following the above-mentioned rule, we simulate nine possibilities, denoted as $iii$, $sss$, $rrr$, $rii$, $iri$, $iir$, $pii$, $ipi$, $iip$.
To obtain identical characteristics (e.g. shapes), we randomly choose (from a uniform distribution) the value of the first pore characteristic (e.g. $\gamma_1=0.1$) and then assign the same value to the rest of the pores. 
To get slightly varying shapes or sizes, we randomly choose the variations up to $10\%$ with respect to the first pore. To obtain slightly varying orientation, we simulate a random rotational axis, and we rotate this axis by angles that vary again up to $10\%$. To get random characteristic, we again use a uniform distribution to draw the characteristic for the first pore, and we repeat such a random simulation $m$ times. A pattern $p$ means that we randomly choose a characteristic, copy its values for $m/3$ pores, draw the value again for other $m/3$ pores, and finally draw the characteristic for the rest of the pores. All procedures described above must be looped for changing $m$. In other words, they are repeated $1000$ times till each possibility of $m$ is furnished. Having simulated certain shapes, $\gamma_c$, orientations, and densities, $e_c$, we use these values to calculate $H_{ijk\ell_c}$ and $\phi_c$. To do so, we follow the effective medium theory summarised by~\citet{Kach18}. Note that dry excess compliances are impacted by shape and orientation, whereas pore volume fraction is influenced by shape and size. Having $H_{ijk\ell_c}$ and $\phi_c$, we use equations from this paper to get $\Delta_I$ and $\Delta_{II}$ that lead to $R^{f}_I$. The goal of the simulations is to confirm the cases when the equivalence~(\ref{twomet}) or the approximation~(\ref{twomet4}) occurs and to show which characteristics affect $R_I^f$ the most.

To perform the tests, apart from the solid phase and fluid compressibility, we also need to define shape and size ranges. 
We choose two different ranges of pore shapes. Range $\gamma_c\in(0, 2)$ considers both oblate and prolate spheroids, whereas $\gamma_c\in(0, 0.2)$ corresponds to crack-like pores only.
We select $e_c\in(0, 3/m)$ so that the maximum volume fraction of a set
\begin{equation}
\phi^{(p)}=\frac{4\pi\gamma_cme_c}{3}
\end{equation}
is equal for any $m$ and can reach around $25\%$ (if $\gamma_c\approx2$). During the simulations, it occurred that the choice of the size range had negligible impact on $R_I^f$.
Let us discuss the results presented in Table~\ref{tab}. We notice that both random and patterned characteristics lead to significant discrepancies that can reach up to $88\%$ ($ppp$). In general, random pores generate a higher mean $R_I^f$, but the results do not vary as much as in the case of patterns. The comparison between $rrr$ and $ppp$ is also shown in Figure~\ref{fig:discr}. It is clear that the number of pores has a negligible impact on $ppp$ but significantly reduces oscillations of $R_I^f$ for $rrr$. Results for $sss$ support the approximation~(\ref{twomet4}). In other words, $R_I^f$ is very low if the pore microstructure varies up to $10\%$ in the set. Looking at $rii$, $iri$, and $iir$ (alternatively, $pii$, $ipi$, $iip$), we can evaluate the impact of each characteristic on $R_I^f$. For instance, $rii$ can indicate the influence of shape, since identical orientations and sizes have no contribution to the discrepancy.
We notice that shape or orientation has a significant impact on $R_I^f$, whereas the effect of size is negligible, as expected from the theoretical considerations (equivalence~(\ref{twomet})). Also, the results depend strongly on the choice of the $\gamma_c$ range, but little on the choice of the solid matrix.
\begin{table}[htb]
\caption{\footnotesize{Chosen scenarios of microstructure with corresponding $R_I^f$ (in $\%$). Generated pores with density $e_c=3/m$ are embedded in a Berea sandstone ($E=87\,\rm{GPa}\,,\,\nu=0.11$) or basalt ($E=125\,\rm{GPa}\,,\,\nu=0.25$). Mean and maximum discrepancies, along with standard deviations, are presented. }}
\label{tab}
\begin{tabular}
{cccccccccc}
& \multicolumn{6}{c}{Berea sandstone} & \multicolumn{3}{c}{Basalt}\\
\toprule
 & \multicolumn{3}{c}{$\gamma\in(0,2)$} & \multicolumn{3}{c}{$\gamma\in(0,0.2)$}  & \multicolumn{3}{c}{$\gamma\in(0,0.2)$}\\
\toprule
case \hphantom{X} & mean & max & sd \hphantom{X} & mean & max & sd \hphantom{X} & mean & max & sd \\
 \cmidrule{1-10}
rrr  \hphantom{X}& 9.71 & 25.76 & 1.77  \hphantom{X}& 49.13 & 70.48 & 1.89 \hphantom{X}&53.78&81.11&1.95 \\
\cmidrule{1-10}
ppp \hphantom{X}& 10.01 & 82.15 & 16.25 \hphantom{X}& 34.76 & 85.94 & 17.94 \hphantom{X}&38.64&88.75&19.72 \\
\cmidrule{1-10}
sss \hphantom{X}& 0.23 & 3.59 & 0.44 \hphantom{X}& 0.95 & 4.09 & 0.78 \hphantom{X}&1.03&4.01&0.83 \\
\cmidrule{1-10}
rii \hphantom{X}& 15.27 & 34.95 & 6.79 \hphantom{X}& 22.46 & 32.05 & 1.69 \hphantom{X}&29.17&50.44&1.69 \\
\cmidrule{1-10}
iri \hphantom{X}& 7.54 & 50.86 & 13.00 \hphantom{X}& 42.40 & 59.81 & 6.31 \hphantom{X}&44.62&58.01& 4.82 \\
\cmidrule{1-10}
iir \hphantom{X}& $<10^{-12}$ & $<10^{-12}$ & $<10^{-13}$ \hphantom{X}& $<10^{-12}$ & $<10^{-12}$ & $<10^{-13}$ \hphantom{X}&$<10^{-12}$ & $<10^{-12}$ & $<10^{-13}$ \\
\cmidrule{1-10}
pii \hphantom{X}& 10.89 & 80.98 & 16.81 \hphantom{X}& 15.37 & 48.99 & 12.10 \hphantom{X}&18.63&61.92&14.73\\
\cmidrule{1-10}
ipi \hphantom{X}& 5.31 & 75.69 & 11.15 \hphantom{X}& 30.31 & 73.00 & 16.51 \hphantom{X}&32.56&71.88&16.90\\
\cmidrule{1-10}
iip \hphantom{X}&  $<10^{-11}$ &  $<10^{-11}$ &  $<10^{-11}$ \hphantom{X}&  $<10^{-11}$ &  $<10^{-11}$ &  $<10^{-11}$ \hphantom{X}&$<10^{-11}$ &  $<10^{-11}$ &  $<10^{-11}$ \\
\bottomrule
\end{tabular}
\end{table}
%
    \begin{figure}[!htbp]
           \begin{floatrow}
             \ffigbox{\includegraphics[scale = 0.44]{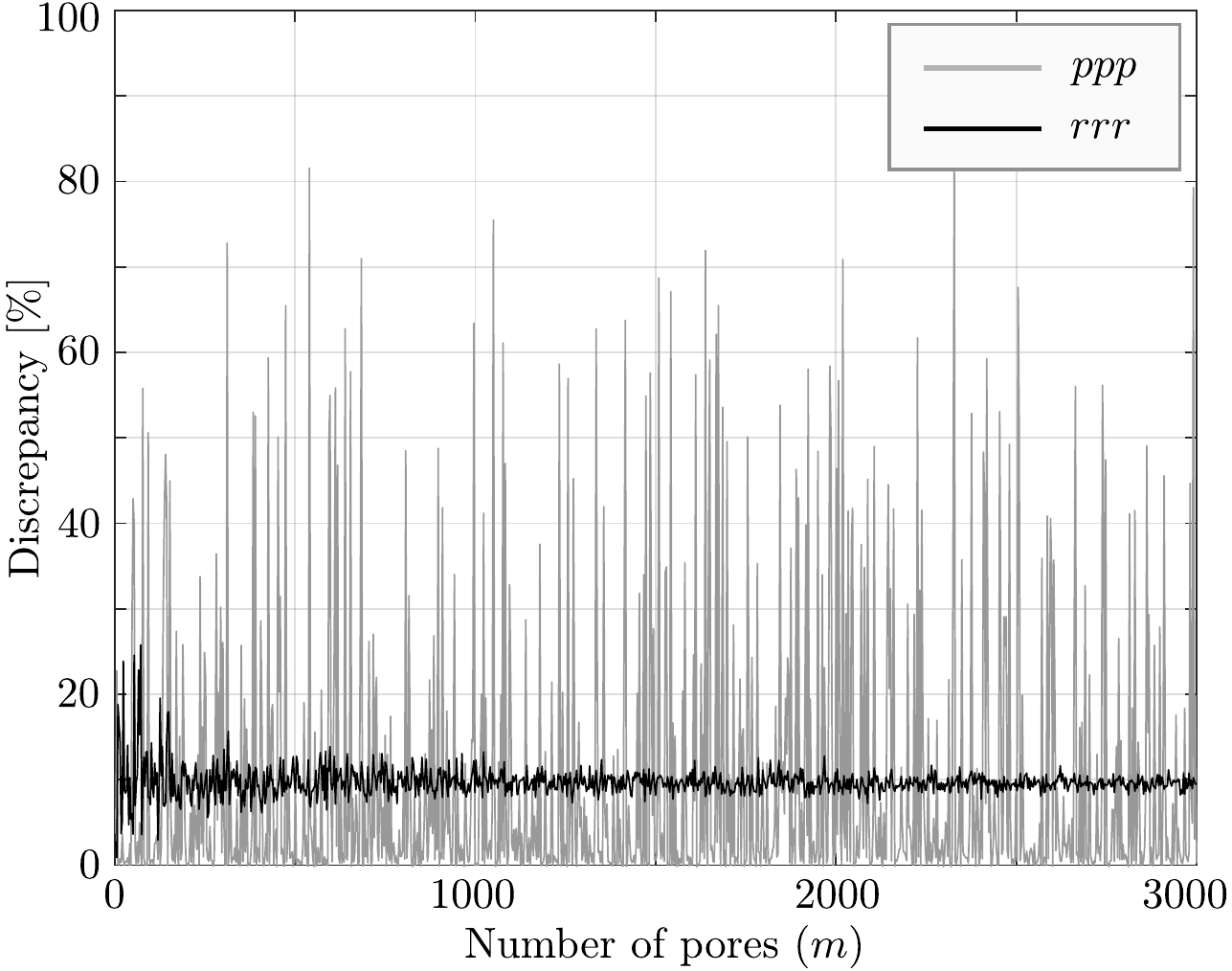}}{\caption{\footnotesize{Discrepancy between pore-impact and set-impact approaches based on multiple simulations for each $m$. Black colour corresponds to the scenario of random geometries ($rrr$), whereas grey indicates geometries with non-random pattern ($ppp$). Berea sandstone, $\gamma_c\in(0, 2)$, and $e_c=3/m$.}}\label{fig:discr}}
         \qquad  
 \ffigbox{\includegraphics[scale = 0.43]{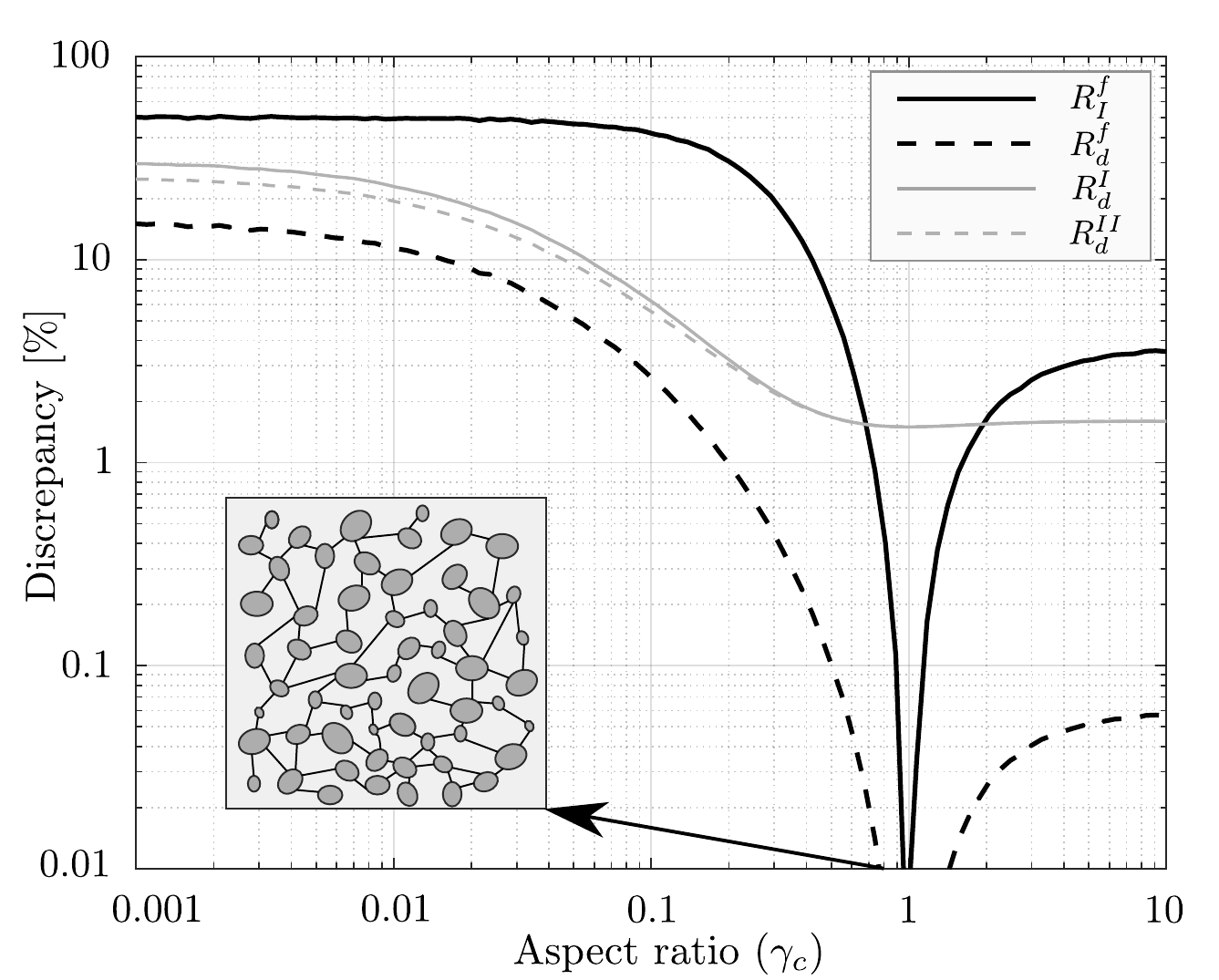}}{\caption{\footnotesize{Aspect ratio vs. discrepancies between various characteristics ($m=100$). Pore orientations and sizes are generated randomly ($irr$). Solid black denotes $R^{f}_I$, dashed black $R^{f}_{d}$, solid grey $R^{I}_{d}$, and dashed grey $R^{II}_{d}$. Berea sandstone and $e_c=3/m$.}\label{fig:gamma}}}
           \end{floatrow}
        \end{figure}

Although the relative discrepancy between fluid effects can be very high, it is good to relate it to the impact of dry pores. Perhaps, there are cases when $R_I^f$ is large but excess compliances of pores are negligible; then, the choice of micromechanical approach would not matter. Hence, we propose to also utilise 
\begin{equation}
R^{f}_d:=\frac{||\Delta_I-\Delta_{II}||}{||\sum_{c=1}^x\phi_cH_c||}\times 100\%
\,, \qquad
R^{I}_d:=\frac{||\Delta_I||}{||\sum_{c=1}^x\phi_cH_c||}\times 100\%
\,,
\qquad
R^{II}_{d}:=\frac{||\Delta_{II}||}{||\sum_{c=1}^x\phi_cH_c||}\times 100\%\,,
\end{equation}
where subscript $d$ denotes the effect of dry pores.
In Figure~\ref{fig:gamma}, we use these additional measures in the context of $\gamma_c$.
Therein, we simulate $irr$ for $m=100$ and show how a fixed aspect ratio affects the discrepancies.
Even tough $R_I^f>10\%$ for $\gamma_c$ up to $\sim0.4$, the choice of $\Delta_I$ or $\Delta_{II}$ is essential only for very low aspect ratios. Note that discrepancy $R_d^f>10\%$ occurs for $\gamma_c$ up to $\sim0.02$. Hence, the choice of fluid effect computation is essential only for cracks. In general, the fluid effect (either $\Delta_I$ or $\Delta_{II}$) is not so important to consider if aspect ratios are not low. $R_d^I>10\%$ or $R_d^{II}>10\%$ occurs for $\gamma_c$ up to $\sim0.04$ or $\sim0.05$, respectively. If we simulate $ipp$ instead of $irr$, the discrepancies occur to be even lower. As discussed in Section~\ref{sec:sphere}, $\gamma_c=1$ implies $\Delta_I=\Delta_{II}$; hence, $R_I^f$ and $R_d^f$ tend to zero for quasi-spherical pores. 

To sum up, our simulations indicate that $R^{f}_I$ depends mostly on the following (in descending order of importance).
\begin{enumerate}
\item{Shape and orientation---the effect is large, especially if the aspect ratio is small.}
\item{Number of pores---only important if shape and orientation are random.}
\item{Solid phase---little impact only.}
\item{Pore size---negligible effect, if any.}
\end{enumerate}
Also, it happens it is not essential to consider the fluid effect in the case of prolate pore shapes. In turn, in general, the choice of the micromechanical approach matters only for the case of cracks. 
\section{Discussion}
The extension to the anisotropic poroelasticity theory presented in this paper can be particularly useful in the case of isolated sets of cracks embedded in the porous solid. Although perfect isolation of the sets can be regarded as an idealisation, the scenario of weak connections between the sets seems not to be unlikely (see e.g., Figure~\ref{fig:pic}). It could lead to pore pressure variations viewed at a shorter time scale. Therefore, we predict that we will encounter the problem of varying pressure during laboratory experiments under undrained or quasi-undrained conditions. We plan to perform such experiments in the near future in order to test this.

Further, one may expect to utilise our extension at the large field scale, where quasi-static, low-frequency conditions are assumed. 
Such conditions are typical for deep reservoirs, seismic measurements or stress regimes in fault zones. The so-called dual-porosity extension (two isolated, isotropic pore-sets) was already considered in the context of flow patterns of fluids during reservoir pump down~\citep{Berryman02}. Following~\citet{Backus62}, the seismic wavelength is sufficiently long to consider thin layers as an effective anisotropic medium. By analogy, we may treat such material as having multiple pore sets isolated from each other but influencing the overall response. Similarly, assuming low-frequency stresses in the fault zone, distinct sets may lead to different fluid content changes.

The micromechanical description of the poroelastic medium can be viewed either holistically (set-impact) or individually (pore-impact). The former approach describes connected pores, where the sole connections have a negligible mechanical impact, but lead to uniform pressure within a set. The latter method assumes the absence of connections between pores that leads to varying pressure at the pore scale (pressure polarization). Effective elasticity computed using both approaches may differ, but such discrepancy matters only in the case of crack-like pores. The set-impact approach can be considered inconsistent with the EMT perspective (meso instead of micro-scale) but always provides the required poroelastic coefficients. We believe that it can be used freely if not largely inconsistent with the pore-impact approach. In the case it is inconsistent, and cracks are considered, we recommend caution; in such scenarios, the correctness of the set-impact approach ought to be verified in laboratory experiments. In the context of conventional triaxial tests indicating varying pressures, we recommend using the induced CTI symmetry that utilises the pore-impact approach.

If the poroelastic coefficients are related to microstructures---by assuming the undrained case---then we conjecture that these coefficients can be implemented in expressions~(1)--(2), where any (not necessarily undrained) scenario is considered. In other words, it is possible that micromechanics can be indirectly linked to intermediate poroelastic states, where time dependence plays a role. The comprehensive, time-dependent considerations beyond expressions~(1)--(2), where varying permeabilities and interset fluid flows are taken into account, can be found in our parallel paper~\citep{AdamusEtAl23a}.
\section{Conclusions}
We have proposed an extension of anisotropic poroelasticity theory, where we regard various scenarios of pore interconnections. 
The original approach of a single interconnected set of inhomogeneities is generalised to numerous sets (porosity clusters) that are isolated from each other.
Fluid content change can vary from set to set; when summed, giving the total fluid content change of the effective medium.
Each set is described by a distinct storage coefficient and Skempton-like tensor. Summed storages lead to a total storage coefficient, whereas Skempton-like coefficients should be regarded as the poroelastic characteristics of a particular set only.
Our idea of the theory extension originates from the concept of the so-called pore pressure polarisation.

Also, we invoke the classic micromechanical description of the fluid effect on pores, where each pore is treated separately (pore-impact approach).
We show that the pore-impact approach can be used successfully to obtain Skempton-like and storage coefficients, only if pores have identical shapes and orientations in a set.
If they have not, then we propose an alternative description (set-impact approach) that always leads to those coefficients. 
This way, a bridge between micromechanics (effective methods) and poroelasticity (extended Biot theory) is provided.
Although, in general, the classic micromechanical description cannot be implemented into Biot theory, it may be used as a fluid-effect reference ($\Delta_I$ vs. $\Delta_{II}$).
The discrepancies between the two approaches are shown numerically.
The choice of the micromechanical approach matters, especially in the case of cracks.
In our considerations, the non-interactive approximation is assumed. 
This can lead to errors in the case of high pore concentrations. 
Therefore, the expressions provided herein should be treated either as relevant for moderate pore densities only or as the basis for future investigations on the interactive cases.

Further, we have proved that TI excess compliance tensor (CTI symmetry) can be obtained by inserting $n\geq3$ vertical (axial) and aligned sets of cracks that are distributed equally around the symmetry axis.
Set-induced symmetry occurs to be a particular case of CTI, where $\beta_{1133}=0$ must be obeyed.
Optionally, the horizontal (radial) set of cracks can be inserted, which does not affect the symmetry conditions.
\section*{Acknowledgements}
This research was supported financially by the NERC grant: ``Quantifying the Anisotropy of Poroelasticity in Stressed Rock'', NE/N007826/1 and NE/T00780X/1.
\bibliographystyle{apa}
\bibliography{LibraryUpperCase}

\begin{thebibliography}{}

\bibitem[\protect\astroncite{Adamus et~al.}{2023}]{AdamusEtAl23a}
Adamus, F.~P., Healy, D., Meredith, P.~G., Mitchell, T.~M., and Stanton-Yonge,
  A. (2023).
\newblock Multi-porous extension of anisotropic poroelasticity: consolidation
  and related coefficients.
\newblock {\em J. Geophys. Res. Solid Earth}, (to be submitted).

\bibitem[\protect\astroncite{Backus}{1962}]{Backus62}
Backus, G.~E. (1962).
\newblock Long-wave elastic anisotropy produced by horizontal layering.
\newblock {\em J. Geophys. Res.}, 67:4427--4440.

\bibitem[\protect\astroncite{Beeler et~al.}{2000}]{Beeler00}
Beeler, N.~M., Simpson, R.~W., Hickman, S.~H., and Lockner, D.~A. (2000).
\newblock Pore fluid pressure, apparent friction, and {C}oulomb failure.
\newblock {\em J. Geophys. Res. Solid Earth}, 105:25533--25542.

\bibitem[\protect\astroncite{Berryman}{2002}]{Berryman02}
Berryman, J.~G. (2002).
\newblock Extension of poroelastic analysis to double-porosity materials: New
  technique in microgeomechanics.
\newblock {\em J. Eng. Mech.}, 128:840--847.

\bibitem[\protect\astroncite{Biot}{1941}]{Biot41}
Biot, M.~A. (1941).
\newblock General theory of three-dimensional consolidation.
\newblock {\em J. Appl. Phys.}, 12:155--164.

\bibitem[\protect\astroncite{Biot}{1956}]{Biot56}
Biot, M.~A. (1956).
\newblock General solutions of the equations of elasticity and consolidation
  for a porous material.
\newblock {\em J. Appl. Mech.}, 78:91--96.

\bibitem[\protect\astroncite{Biot}{1962}]{Biot62}
Biot, M.~A. (1962).
\newblock Mechanics of deformation and acoustic propagation in porous media.
\newblock {\em J. Appl. Phys.}, 33:1482--1498.

\bibitem[\protect\astroncite{Chen et~al.}{2018}]{Chen18}
Chen, F., Sevostianov, I., Giraud, A., and Grgic, D. (2018).
\newblock Combined effect of pores concavity and aspect ratio on the elastic
  properties of a porous material.
\newblock {\em Int. J. Solids Struct.}, 134:161--172.

\bibitem[\protect\astroncite{Cheng}{1997}]{Cheng97}
Cheng, A. H.-D. (1997).
\newblock Material coefficients of anisotropic poroelasticity.
\newblock {\em Int. J. Rock Mech. Min. Sci.}, 34:199--205.

\bibitem[\protect\astroncite{Cocco and Rice}{2002}]{Cocco02}
Cocco, M. and Rice, J.~R. (2002).
\newblock Pore pressure and poroelasticity effects in {C}oulomb stress analysis
  of earthquake interactions.
\newblock {\em J. Geophys. Res.}, 107:1--17.

\bibitem[\protect\astroncite{Dormieux et~al.}{2006}]{DormieuxEtAl02}
Dormieux, L., Kondo, D., and Ulm, F.-J. (2006).
\newblock {\em Microporomechanics}.
\newblock Wiley.

\bibitem[\protect\astroncite{Eshelby}{1957}]{Eshelby57}
Eshelby, J.~D. (1957).
\newblock The determination of the elastic field of an ellipsoidal inclusion,
  and related problems.
\newblock {\em Proc. R. Soc. A}, 241:376--396.

\bibitem[\protect\astroncite{Geertsma}{1957}]{Geertsma57}
Geertsma, J. (1957).
\newblock The effect of fluid pressure decline on volumetric changes of porous
  rocks.
\newblock {\em Pet. Trans. AIME}, 210:331--340.

\bibitem[\protect\astroncite{Grechka and Kachanov}{2006}]{Grechka06c}
Grechka, V. and Kachanov, M. (2006).
\newblock Effective elasticity of fractured rocks: A snapshot of the work in
  progress.
\newblock {\em Geophys.}, 71:W45--W58.

\bibitem[\protect\astroncite{Guéguen and Sarout}{2009}]{Gueguen09}
Guéguen, Y. and Sarout, J. (2009).
\newblock Crack-induced anisotropy in crustal rocks: Predicted dry and
  fluid-saturated {T}homsen's parameters.
\newblock {\em Phys. Earth Planet. Inter.}, 172:116--124.

\bibitem[\protect\astroncite{Hart and Wang}{1995}]{Hart95}
Hart, D.~J. and Wang, H.~F. (1995).
\newblock Laboratory measurements of a complete set of poroelastic moduli for
  {B}erea sandstone and {I}ndiana limestone.
\newblock {\em J. Geophys. Res.}, 100:17741--17751.

\bibitem[\protect\astroncite{Herman}{1945}]{Herman45}
Herman, B. (1945).
\newblock Some theorems of the theory of anisotropic media.
\newblock {\em Comptes Rendus (Doklady) de l'Académie des Sciences de l'URSS},
  48:89--92.

\bibitem[\protect\astroncite{Ji et~al.}{2010}]{Ji10}
Ji, S., Sun, S., Wang, Q., and Marcotte, D. (2010).
\newblock Lamé parameters of common rocks in the {E}arth's crust and upper
  mantle.
\newblock {\em J. Geophys. Res.}, 115(B6):1--15.

\bibitem[\protect\astroncite{Kachanov}{1980}]{Kachanov80}
Kachanov, M. (1980).
\newblock Continuum model of medium with cracks.
\newblock {\em J. Eng. Mech.}, 106:1039--1051.

\bibitem[\protect\astroncite{Kachanov and Sevostianov}{2018}]{Kach18}
Kachanov, M. and Sevostianov, I. (2018).
\newblock {\em Micromechanics of Materials, with Applications}.
\newblock Springer.

\bibitem[\protect\astroncite{Lockner and Stanchits}{2002}]{Lockner02}
Lockner, D.~A. and Stanchits, S.~A. (2002).
\newblock Undrained poroelastic response of sandstones to deviatoric stress
  change.
\newblock {\em J. Geophys. Res. Solid Earth}, 107:1--14.

\bibitem[\protect\astroncite{Mehrabian}{2018}]{Mehrabian18}
Mehrabian, A. (2018).
\newblock The poroelastic constants of multiple-porosity solids.
\newblock {\em Int. J. Eng. Sci.}, 132:97--104.

\bibitem[\protect\astroncite{Mehrabian and Abousleiman}{2014}]{Mehrabian14}
Mehrabian, A. and Abousleiman, N.~A. (2014).
\newblock Generalized {B}iot's theory and {M}andel's problem of
  multiple-porosity and multiple-permeability poroelasticity.
\newblock {\em J. Geophys. Res. Solid Earth}, 119:2745--2763.

\bibitem[\protect\astroncite{O'Connell and Budiansky}{1974}]{OConnell74}
O'Connell, R.~J. and Budiansky, B. (1974).
\newblock Seismic velocities in dry and saturated cracked solids.
\newblock {\em J. Geophys. Res.}, 79:5412--5426.

\bibitem[\protect\astroncite{O'Connell and Budiansky}{1977}]{OConnell77}
O'Connell, R.~J. and Budiansky, B. (1977).
\newblock Viscoelastic properties of fluid-saturated cracked solids.
\newblock {\em J. Geophys. Res.}, 82:5719--5736.

\bibitem[\protect\astroncite{Oda}{1986}]{Oda86}
Oda, M. (1986).
\newblock Fabric tensor for discontinuous geological materials.
\newblock {\em Soils Found.}, 22:96--108.

\bibitem[\protect\astroncite{Rizzo et~al.}{2018}]{Rizzo18}
Rizzo, R.~E., Healy, D., Heap, M.~J., and Farrell, N.~J. (2018).
\newblock Detecting the onset of strain localization using two‐dimensional
  wavelet analysis on sandstone deformed at different effective pressures.
\newblock {\em J. Geophys. Res. Solid Earth}, 123:460--478.

\bibitem[\protect\astroncite{Sayers and Kachanov}{1995}]{Sayers95}
Sayers, C. and Kachanov, M. (1995).
\newblock Microcrack-induced elastic wave anisotropy of brittle rocks.
\newblock {\em J. Geophys. Res.}, 100:4149--4156.

\bibitem[\protect\astroncite{Schoenberg and Helbig}{1997}]{Sch97}
Schoenberg, M. and Helbig, K. (1997).
\newblock Orthorhombic media: Modeling elastic wave behavior in a vertically
  fractured earth.
\newblock {\em Geophys.}, 62:1954--1974.

\bibitem[\protect\astroncite{Schoenberg and Sayers}{1995}]{Sch95}
Schoenberg, M. and Sayers, C.~M. (1995).
\newblock Seismic anisotropy of fractured rock.
\newblock {\em Geophys.}, 60:204--211.

\bibitem[\protect\astroncite{Shafiro and Kachanov}{1997}]{Shafiro97}
Shafiro, B. and Kachanov, M. (1997).
\newblock Materials with fluid-filled pores of various shapes: Effective
  elastic properties and fluid pressure polarization.
\newblock {\em Int. J. Solids Struct.}, 34:3517--3540.

\bibitem[\protect\astroncite{Wong}{2017}]{Wong17}
Wong, T.~F. (2017).
\newblock Anisotropic poroelasticity in a rock with cracks.
\newblock {\em J. Geophys. Res. Solid Earth}, 122:1--15.

\end{thebibliography}
\numberwithin{equation}{section}
\appendix
\section{List of symbols}\label{sec:list}
\begin{table}[!htbp]
\scalebox{1}{
\begin{tabular}
{l l l c l l l}
\multicolumn{7}{l}{Greek letters}\\ 
[1ex]
\multicolumn{3}{l}{Scalars} & &               \multicolumn{3}{l}{Tensors}\\
[1ex]
$\gamma$&:=& aspect ratio                  & &                            $\alpha_{ij}$&:=& $2^{\rm{nd}}$ rank crack density tensor\\
$\delta$&:=& fluid factor                  & &                                 $\beta_{ijk\ell}$&:=& $4^{\rm{th}}$ rank crack density tensor\\
$\zeta$&:=& fluid content change                  & &                                    $\Delta_{ijk\ell}$&:=& fluid effect\\
$\theta$&:=& angle between vertical crack sets                 & &                         $\delta_{ij}$&:=& Kronecker delta\\
$\nu$&:=& Poisson ratio of a solid phase       & &         $\varepsilon_{ij}$&:=& strain tensor\\
$\phi$&:=& volume fraction             & &                   $\sigma_{ij}$&:=& stress tensor\\
$\varphi$&:=& azimuthal angle             & &                   \\
$\psi$&:=& non-zero angle             & &                   \\
[3ex] 
%
\multicolumn{7}{l}{Roman letters}\\ 
[1ex]
\multicolumn{3}{l}{Scalars} & &               \multicolumn{3}{l}{Tensors}\\
[1ex]
$I$&:=& pore impact approach& &                                        $B_{ij}$&:=& Skempton tensor\\
$II$&:=& set impact approach& &                                         $H_{ijk\ell}$&:=&dry excess compliance tensor\\
$a$&:=& ratio of a circular crack       & &                               $\Delta H_{ijk\ell}$&:=& saturated compliance tensor\\
$c$&:=& closed or connected pore& &                                   $n_{i}$&:=& normal to crack surface\\
$d$&:=& dry pore& &                                                          $Q_{ij}$&:=& pressure polarisation tensor\\
$E$&:=& Young modulus of a solid phase& &                        $S_{ijkl}$&:=& compliance tensor of a porous skeleton\\
$e_c$&:=& density of a single crack& &                                 $S^0_{ijkl}$&:=& compliance tensor of a solid phase\\
$K_0$&:=& bulk modulus of a solid phase& &                          & & \\
$K_d$&:=& bulk modulus of a dry pore& &                              & & \\
$K_f$&:=& bulk modulus of a fluid phase& &                           & & \\
$m$&:=& number of pores in a set (or subset)& &                   & & \\
$n$&:=& number of sets (or subsets)& &                                & & \\
$p$&:=& particular set (or subset)& &                                    & & \\
$p_f$&:=& pore pressure& &                                                 & & \\
$R$&:=& relative discrepancy (error)& &                                 & & \\
$S$&:=& storage coefficient& &                                             & & \\
$u$&:=& undrained entity& &                                                & & \\
$V$&:=& medium's volume& &                  & & \\
$x$&:=& total number of pores in a medium & &                             & & \\
$x_i$&:=& coordinate axis & &                                              & & \\
$y$&:=& unknown or constant  & &                                        & & \\
$Z$&:=& crack excess compliance& &                                    & & \\
\end{tabular}
}
\end{table}
\section{Properties of poroelastic parameters}\label{sec:ap}
Let us discuss some key properties of $S^{(p)}$, $B_{ij}^{(p)}$, $p_f^{(p)}$, and $\zeta^{(p)}$; poroelastic parameters describing distinct sets at a mesoscopic scale. Herein, we show how these parameters relate to their bulk counterparts describing a medium with single (instead of multiple) porosity, which is the original Biot's case.

In undrained conditions, the storage coefficient is not dependent on pore pressure---which must be constant in the set but may vary in the medium. 
Therefore, this coefficient does not have to be linked strictly with a pore set---as is the case of the Skempton-like tensor or obviously the aforementioned pressure---but can also be viewed at other, single-pore or bulk-medium scales. In other words, the storage coefficient is a scalar independent of the connection among pores. Thus, we can use definition~(\ref{def:storage}) along with~(\ref{Kc}) to express a total storage coefficient, $S_{tot}$, as
\begin{equation}\label{a1}
S_{tot}=\sum_{p=1}^nS^{(p)}=\sum_{p=1}^n\phi^{(p)}K_d^{{(p)}^{-1}}+\sum_{p=1}^n\phi^{(p)}\left(K_f^{-1}-K_0^{-1}\right)=\sum_{p=1}^n\left(\phi^{(p)}\sum_{i=1}^3\sum_{j=1}^3H_{iijj}^{{(p)}}\right)+\sum_{p=1}^n\phi^{(p)}\left(K_f^{-1}-K_0^{-1}\right)\,.
\end{equation}
Let us refer to expression~(\ref{def:Hp0}) to analogously define the relationship between dry excess compliances of multiple sets and of the entire effective medium, $H_{ijk\ell}$. We can write
\begin{equation}\label{newdef}
\phi_{tot}H_{ijk\ell}=\sum_{p=1}^n\phi^{(p)}H_{ijk\ell}^{(p)}\,.
\end{equation}
Using the above, we can rewrite~(\ref{a1}) as
\begin{equation}
S_{tot}=\phi_{tot}\left(\sum_{i=1}^3\sum_{j=1}^3H_{iijj}+K_f^{-1}-K_0^{-1}\right)=\phi_{tot}\left(K_d^{-1}+K_f^{-1}-K_0^{-1}\right)=:S\,.
\end{equation}
We see that $S_{tot}$ is equivalent to the storage coefficient of a medium with one interconnected porosity, $S$.

In contrast, the Skempton-like tensor and pore pressure must be associated with pore connections. 
They are linked to each set and must be regarded as that; therefore, a total value of the above-mentioned parameters is not introduced.
Nevertheless, one may again seek a comparison between a bulk Skempton tensor (or pressure) calculated for a medium with one interconnected porosity and set Skemptons (or pressures) obtained for the same medium but with detached sets.
Bulk, interconnected Skempton can be expressed as
\begin{equation}
B_{ij}:=\dfrac{3\sum_{k=1}^3H_{ijkk}}{K_d^{{-1}}+K_f^{-1}-K_0^{-1}}=\dfrac{3\phi_{tot}\sum_{k=1}^3H_{ijkk}}{S}=\dfrac{3\sum_{p=1}^n\left(\phi^{(p)}\sum_{k=1}^3H_{ijkk}^{(p)}\right)}{\sum_{p=1}^nS^{(p)}}\,,
\end{equation}
where we utilised relationship~(\ref{newdef}). If $K_d^{(p)}$ is constant through sets, then $B_{ij}$ reduces to
\begin{equation}
B_{ij}=\dfrac{3\sum_{p=1}^n\left(\phi^{(p)}\sum_{k=1}^3H_{ijkk}^{(p)}\right)}{\phi_{tot}\left(K_d^{{(p)}^{-1}}+K_f^{-1}-K_0^{-1}\right)}=\frac{\sum_{p=1}^n\phi^{(p)}B_{ij}^{(p)}}{\phi_{tot}}=\overline{B_{ij}^{(p)}}\,,
\end{equation}
where bar denotes an average weighted by the volume fraction of each pore set.
Due to relation~(\ref{three}), the same conclusions regard interconnected pore pressure $p_f$ and set pressures $p_f^{(p)}$---we obtain $p_f=\overline{p_{f}^{(p)}}$ if $K_d^{(p)}=\rm{const}$. 

In drained conditions, total fluid content change (additive scalar) can be written as
\begin{equation}
\zeta_{tot}=\sum_{p=1}^n\zeta^{(p)}=\frac{1}{3}\sum_{p=1}^n\left(S^{(p)}\sum_{k=1}^3\sum_{\ell=1}^3B^{(p)}_{k\ell}\sigma_{k\ell}\right)=
\frac{n}{3}\sum_{k=1}^3\sum_{\ell=1}^3{\overline{S^{(p)}B^{(p)}_{k\ell}}}\sigma_{k\ell}\,.
\end{equation}
Herein, the bar indicates an arithmetic average. If $S^{(p)}$ is constant through sets---that is equivalent to constant both $\phi^{(p)}$ and $K_d^{(p)}$---we get
\begin{equation}
\zeta_{tot}=\frac{n}{3}S^{(p)}\sum_{k=1}^3\sum_{\ell=1}^3\overline{B^{(p)}_{k\ell}}\sigma_{k\ell}
=\frac{1}{3}S\sum_{k=1}^3\sum_{\ell=1}^3B_{k\ell}\sigma_{k\ell}
=:\zeta
\end{equation}
that is a fluid content change of a single interconnected porosity.

Note that $K_d^{(p)}$ depends on the pore shape only.
Using definitions~(\ref{Kc}) and~(\ref{def:Hp}), we can write 
\begin{equation}
K_d^{(p)}=\left(\sum^3_{i=1}\sum^3_{j=1}\frac{\sum_{c=1}^m\phi_cH_{iijj_c}}{\sum_{c=1}^m\phi_c}\right)^{-1}\,,
\end{equation}
By dry excess compliance definition~\citep{Kach18}, each component $H_{ijk\ell_c}$ depends on aspect ratio and pore orientation. However, due to the component summation, the tensor (and the pore) orientation does not matter. Thus, from set to set
\begin{equation}
K_d^{(p)}=\rm{const}\implies \frac{\sum_{c=1}^m\phi_c\gamma_c}{\sum_{c=1}^m\phi_c}=\rm{const}\,.
\end{equation}
To conclude, $B_{ij}$ (or $p_f$) is weighted average of $B_{ij}^{(p)}$ (or $p_f^{(p)}$) and $\zeta=\zeta_{tot}$ if each set has the same composition of pore shapes. In a typical geological scenario, pore shapes vary from set to set ($K_d^{(p)}\neq\rm{const}$) that leads to $S=S_{tot}$, $p_f\neq\overline{p_{f}^{(p)}}$, $B_{ij}\neq\overline{B_{ij}^{(p)}}$, and $\zeta\neq\zeta_{tot}$.

\section{Proof of CTI induction for $n\geq3$ vertical sets}\label{sec:ap1}
\begin{theorem}\label{thm}
Consider $n\geq3$ $(n\in\mathbb{N})$ identical sets of aligned and dry circular cracks that are embedded in the isotropic solid phase. If each crack set is vertical and isolated from the other by horizontal angle $\theta=\pi/n$, then transverse-isotropy (TI) with a vertical symmetry axis is induced.  
\end{theorem}
\begin{proof}
Without loss of generality, assume that $x_3$ is a vertical axis. 
Vertical cracks in each set are aligned, meaning that their surface normals are equal to the normal of the $p$ set,
\begin{equation}
\bm{n}^{(p)}=[\cos \left(p\theta\right),\,\sin \left(p\theta\right),\,0]\,.
\end{equation}
A TI excess compliance matrix, $\bm{H}$, is a sufficient condition for the effective medium to become TI.
An excess compliance matrix is TI with a vertical symmetry axis if components of crack density tensors
\begin{equation}\label{alphabeta2}
\alpha_{11}=\alpha_{22}\,,\qquad\beta_{1111}=\beta_{2222}\,,\qquad\beta_{1111}=3\beta_{1122}\,.
\end{equation}
The above conditions define a particular case of TI symmetry, the so-called cylindrical transverse-isotropy (CTI).
If each set is identical, meaning that the number, sizes, and shapes of cracks are the same in each set, we can rewrite~(\ref{alphabeta})--(\ref{alphabetab}) as
\begin{align}
\alpha_{ij}&:=\sum_{c=1}^mZ_{T_c}n_{i_c}n_{j_c}=Z_T\sum_{p=1}^nn_i^{(p)}n_j^{(p)}\,,\\
\beta_{ijk\ell}&:=\sum_{c=1}^m\left(Z_{N_c}-Z_{T_c}\right)n_{i_c}n_{j_c}n_{k_c}n_{\ell_c}=(Z_N-Z_T)\sum_{p=1}^nn_i^{(p)}n_j^{(p)}n_k^{(p)}n_\ell^{(p)}\,.
\end{align}
Thus, conditions~(\ref{alphabeta2}) correspond to
\begin{equation}\label{cos_conditions}
\sum_{p=1}^n\cos^2(p\theta)=\sum_{p=1}^n\sin^2(p\theta)\,,\qquad\sum_{p=1}^n\cos^4(p\theta)=\sum_{p=1}^n\sin^4(p\theta)\,,\qquad\sum_{p=1}^n\cos^4(p\theta)=3\sum_{p=1}^n\cos^2(p\theta)\sin^2(p\theta)\,.
\end{equation}
To prove set-induced CTI symmetry, we need to show that conditions~(\ref{cos_conditions}) are satisfied.
\begin{lemma}
If $n$ sets from Theorem~\ref{thm} are embedded in the solid phase, then conditions~(\ref{cos_conditions}) are satisfied.  
\end{lemma}
\begin{proof}
Lagrange trigonometric identity states that
\begin{equation}
\sum_{p=1}^{n}\cos(p\psi)=-\frac{1}{2}+\frac{\sin\left[\left(n+\frac{1}{2}\right)\psi\right]}{2\sin\left(\frac{\psi}{2}\right)}\,
\end{equation}
for $\psi\neq2k\pi$ ($k=0,1,2,\dots$).
Therefore, if $\psi=2\theta=2\pi/n$, we obtain
\begin{equation}\label{2theta}
\sum_{p=1}^{n}\cos(2p\theta)=0\,.
\end{equation}
If $\psi=4\theta=4\pi/n$, we get
\begin{equation}\label{4theta}
\sum_{p=1}^{n}\cos(4p\theta)=0\,
\end{equation}
that does not hold for $n=2$.
We can rewrite relation~(\ref{2theta}) as
\begin{equation}
\begin{aligned}
\sum_{p=1}^{n}\cos^2(p\theta)&=\frac{n}{2}\,
\\
\sum_{p=1}^{n}\cos^2(p\theta)&=n-\sum_{p=1}^{n}\cos^2(p\theta)\,
\\
\sum_{p=1}^{n}\cos^2(p\theta)&=\sum_{p=1}^{n}\left[1-\cos^2(p\theta)\right]\,
\\
\sum_{p=1}^{n}\cos^2(p\theta)&=\sum_{p=1}^{n}\sin^2(p\theta)\,
\end{aligned}
\end{equation}
or as
\begin{equation}
\begin{aligned}
\sum_{p=1}^{n}\cos^4(p\theta)&=n-2\sum_{p=1}^{n}\cos^2(p\theta)+\sum_{p=1}^{n}\cos^4(p\theta)\,
\\
\sum_{p=1}^{n}\cos^4(p\theta)&=\sum_{p=1}^{n}\left[1-\cos^2(p\theta)\right]^2\,
\\
\sum_{p=1}^{n}\cos^4(p\theta)&=\sum_{p=1}^{n}\sin^4(p\theta)\,.
\end{aligned}
\end{equation}
Assume that the last condition is satisfied, namely,
\begin{equation}\label{condition3}
\sum_{p=1}^n\cos^4(p\theta)=3\sum_{p=1}^n\cos^2(p\theta)\sin^2(p\theta)\,
\end{equation}
that can be rewritten as
\begin{equation}
\begin{aligned}
4\sum_{p=1}^{n}\cos^4(p\theta)&=3\sum_{p=1}^{n}\cos^2(p\theta)
\\
n+2\sum_{p=1}^{n}\cos(2p\theta)+\sum_{p=1}^{n}\cos^2(2p\theta)&=3\sum_{p=1}^{n}\cos^2(p\theta)\,
\\
\sum_{p=1}^{n}\cos(2p\theta)+\sum_{p=1}^{n}\cos(4p\theta)&=0\,.
\end{aligned}
\end{equation}
The assumption~(\ref{condition3}) must be correct for $n\geq3$ due to relations~(\ref{2theta}) and~(\ref{4theta}).
\end{proof}
\end{proof}
\end{document}